\documentclass[10pt]{article}
\usepackage{pslatex,fancyhdr,verbatim,calc,algorithm}
\usepackage{graphicx}
\usepackage[page]{appendix}
\usepackage[subrefformat=parens,labelformat=parens]{subfig}
\usepackage{amsthm}
\usepackage{amssymb}
\usepackage{amsmath}
\usepackage{algorithm}
\usepackage{float}
\usepackage{enumerate}
\usepackage{amsfonts}
\usepackage{multicol}
\usepackage{extarrows}
\usepackage{verbatim}
\usepackage{framed}
\usepackage[noend]{algpseudocode}
\usepackage{sectsty}
\usepackage{etoolbox}
\usepackage{tabularx}

\newtheorem{defin}{Definition}
\newtheorem{rem}{Remark}
\theoremstyle{plain}
\newtheorem{theorem}{Theorem}
\newtheorem{assum}{Assumption}

\newtheorem{lemma}{Lemma}

\newcommand{\tabitem}{~~\llap{\textbullet}~~}

\allowdisplaybreaks
\title{Secure and trusted white-box verification} 

\author{
Yixian Cai, George Karakostas, Alan Wassyng \\
McMaster University, \\
Dept. of Computing and Software, \\
1280 Main St. West,
Hamilton, Ontario L8S 4K1, Canada, \\
      {\tt \{caiy35,karakos,wassyng\}$@$mcmaster.ca}}

\begin{document}

\maketitle

\begin{abstract}
Verification is the process of checking whether a product has been implemented according to its prescribed specifications.
We study the case of a designer (the developer) that needs to verify its design by a third party (the verifier),
by making publicly available a limited amount of information about the design, namely a diagram of interconnections
between the different design components, but not the components themselves or the intermediate
values passed between components. We formalize this notion of limited information using tabular expressions
as the description method for both the specifications and the design. Treating verification as a process akin to testing,
we develop protocols that allow for the design to be verified on a set of inputs generated by the verifier, using any test-case
generating algorithm that can take advantage of this extra available information (partially white-box testing), and without
compromising the developer's secret information. Our protocols work with both trusted and untrusted developers, and
allow for the checking of the correctness of the verification process itself by any third party, and at any time. 
\end{abstract}

\section{Introduction}

The following scenario is a rather common example of the design and development of a product: A car manufacturer needs
to incorporate an engine component, made by a third subcontracting party, into its car design. In order to do that, it first 
produces the {\em specifications} of the component (described in a standard way), and announces them to 
the subcontractor. Then the latter implements its own design, and delivers the component, 
claiming that it complies with the specs. The car manufacturer, in turn, {\em verifies} that the claim is true,
by running a set of {\em test cases} on the component, and confirming that the outputs produced 
for these test inputs are consistent
with the specs. After the car design has been completed, a government agency needs to approve it, by running its own
tests and confirming that the outputs comply with its own publicly-known set of environmental specs; once this happens,
the car can be mass-produced and is allowed to be on the road. The only problem is that both
the car manufacturer and the subcontractor {\em do not want to reveal} their designs to any third party, since they contain
proprietary information, such as the setting of certain parameters that took years of experimentation to fine tune; therefore, 
and with only the specs {\em publicly} available, only black-box testing can be performed by the car manufacturer (vis-a-vis the 
subcontractor's component), and the government agency (vis-a-vis the car). It may be the case, though, that the car manufacturer
and/or the subcontractor are willing to reveal some partial information about their implementation(s) (possibly information
that any good mechanic would figure out anyways given some time, such as which part is connected to some other part). 
In this case, an obvious question is the following: "Can we allow for testing, that possibly
takes advantage of this extra information, {\em without revealing any information beyond that}?" The answer to this question
is rather obvious for the two extremes of (i) completely hidden information (only {\em black-box} testing is possible), 
and (ii) completely
revealed information ({\em completely white-box} testing is possible). In this work, we present a method of answering the question 
in the affirmative in the continuum between these two extremes, i.e., for partially white-box testing. 

More specifically, we are going to restrict our exposition to the case of Software Engineering, i.e., to software products, for 
clarity reasons, although our methods apply to any design whose specs and component interaction can be described using
a standard description method, like tabular expressions \cite{alspaugh1992software}. 
There are two parties, i.e., a {\em software developer} (or 
{\em product designer} - we will use these terms interchangeably in what follows), and a {\em verifier}, as well as a publicly
known set of {\em specifications} (or {\em specs}), also in a standard description. The two parties engage in a protocol, whose
end-result must be the {\em acceptance} or {\em rejection} of the design by the verifier, and satisfies the following two properties:
\begin{itemize}
\item {\bf Correctness:} If the two parties follow the protocol, and all test cases produced and tried by the verifier comply 
with the specs, then the verifier accepts, otherwise it rejects.
\item {\bf Security:} By following the protocol, the developer doesn't reveal any bits of information in addition to the information 
that is publicly available.
\end{itemize}
In this setting, the specifications may come from the developer, the verifier, or a third party, but are always public knowledge.
The algorithm that produces the test cases tried by the verifier is run by the latter, and is also publicly known and chosen
ahead of time. Depending on the application, the developer may be either {\em trusted} ({\em honest}), i.e., follows the 
protocol exactly, providing always the correct (encrypted) replies to the verifier's queries, or {\em untrusted} ({\em dishonest})
otherwise, and similarly for the verifier.    

In this work, first we give a concrete notion of {\em partial} white-box testing, using tabular expressions as the description
tool of choice. In its simpler form, a tabular expression is essentially a table with rows, with each row containing a left-hand
side predicate on the table inputs, and a right-hand computation that produces the table outputs. The RHS of a row is
executed iff the LHS predicate is evaluated to {\sc TRUE}. A table describes the functionality of a design component; the
description of the whole design is done by constructing a {\em directed acyclic table graph}, 
representing the interconnection of tables
(nodes of the graph) as the output of one is used as an input to another (edges of the graph). Our crucial assumption is
that the table graph is exactly the {\em extra publicly-known information} made available by the developer (in addition to the
single publicly-known table describing the specifications); the contents
of the tables-nodes, as well as the actual intermediate table input/output values should remain {\em secret}. Our goal is
to facilitate the testing and compliance verification of designs that are comprised by many components, which may be 
off-the-shelf or implemented by third parties; if these components are trusted to comply with their specs by the developer 
(by, e.g., passing a similar verification process), then all the developer needs to know, in order to proceed with the
verification of the whole design, is their input/output functionality, and not the particulars of their implementation. The formal
definition of table graphs can be found in Section \ref{sub:tables}.
   
We present protocols that allow the verifier to run test case-generating algorithms that may take advantage of this
extra available information (e.g., MC/DC \cite{hayhurst2001practical}), and satisfy the correctness 
and security properties (formally defined in
Section \ref{evaluate}) with high probability (whp). We break the task of verification in a set of algorithms for the
encryption of table contents and intermediate inputs/outputs (run by the developer), and an algorithm for checking the
validity of the verification (cf. Definition \ref{defe}); the former ensure the security property, while the latter will force
the verifier to be honest, and allow any third party to verify the correctness of the verification process at any future time. 
Initially we present a protocol for the case of a trusted developer (Section \ref{sec:honest}), 
followed by a protocol for the case of a dishonest one (Section \ref{sec:maldev}). Our main cryptographic tool is
\emph{Fully Homomorphic Encryption} (FHE), a powerful encryption concept that allows computation over 
encrypted data first implemented in~\cite{gentry2009fully}. We also employ bit-commitment protocols \cite{naor1991bit}
in Section \ref{sec:maldev}. 

{\bf Previous work} The goal of {\em obfuscation}, i.e., hiding the code of a program, while maintaining its original functionality,
is very natural, and has been the focus of research for a long time. An obfuscated program reveals no information about the 
original program, other than what can be figured out by having black-box access to the original program. There are many 
heuristic obfuscation methods, e.g., \cite{collberg1997}, \cite{wang2000software}. However, as shown in \cite{barak2001}, 
an obfuscation algorithm that strictly satisfies the definition of obfuscation does not exist. Hence, all obfuscation algorithms 
can only achieve obfuscation to the extent of making obfuscated programs hard to reverse-engineer, while non black-box
information about the original program cannot be guaranteed to be completely secret. 

A crucial aspect of our work is the exclusion of any third-parties that act as authenticators, or guarantors, or
a dedicated server controlled by both parties like the one in \cite{chaki13}, called an `amanat'. In other words, 
we require that there is no shared resource between the two interacting parties, except the communication channel
and the public specs. This reduces the opportunities for attacks (e.g., a malicious agent taking over 
a server, or a malicious guarantor), since they are reduced to attacking the channel, a very-well and studied problem.   

Though ideal obfuscation is impossible to achieve, there are still some other cryptographic primitives that can be used to 
protect the content of a program. One is {\em point function obfuscation} \cite{lynn2004positive}, which allows for the 
obfuscation of a point function (i.e., a function that outputs 1 for one specific input, and 0 otherwise), but not arbitrary functions. 
Another relevant cryptographic primitive is Yao's \emph{garbled circuits} \cite{yao1982protocols}, which allow the computation
of a circuit's output on a given input, without revealing any information about the circuit or the input. Yao's garbled circuits 
nearly achieve what obfuscation requires in terms of one-time usage, and are the foundation of work like one-time programs 
\cite{goldwasser2008one}. Its problem is that the circuit encryption (i.e., the garbled circuit) can be run only once without
compromising security. Recently, \cite{goldwasser2013reusable} proposed a new version of garbled circuits called 
{\em reusable garbled circuits}, which allows for a many-time usage on the same garbled circuit, and still satisfies the 
security of Yao's garbled circuits. \cite{goldwasser2013reusable} then uses reusable garbled circuits to implement 
{\em token-based obfuscation}, i.e., obfuscation that allows a program user to run the obfuscated program on any input, 
without having to ask the developer who obfuscated the program to encode the input before it is used. Token-based obfuscation
guarantees that only black-box information about the original program can be figured out from the obfuscated program. 
Unfortunately, when a program is token-based obfuscated, it becomes a black box to the user (due to the inherent nature of obfuscation), and thus precludes its use for any kind of white-box verification. Building on \cite{goldwasser2013reusable},
our work proposes a method to alleviate this weakness.
   
Plenty of work has been done in the field of {\em verifiable computing}, which is also relevant to our work. Verifiable computing allows one party (the prover) to generate the verifiable computation of a function, whose correctness a second party (a verifier - 
not related to the verifier in our setting) can then verify. For example, a naive way to do this, is to ask the verifier to repeat the
computation of the function; however, in many cases this solution is both inefficient and incorrect. Recently, works like 
\cite{cormode2012practical}, \cite{thaler2012verifiable}, \cite{vu2013hybrid}, \cite{setty2012making}, \cite{setty2012taking},
\cite{parno2013pinocchio}, \cite{ben2013snarks} which are based on the PCP theorem \cite{arora1998probabilistic}, presented
systems that allow the verifier to verify the computation result without re-executing the function. The difference between verifiable computing and our work is that the developer wishes to hide the implementation of a design, while for verifiable computing, 
the design implementation cannot be hidden. However, verifiable computing still has the potential to be applied to the
implementation of secure and trusted verification.


\section{Preliminaries}  \label{sec:prelim}

\subsection{Tabular Expressions (Tables)}  \label{sub:tables}

\emph{Tabular expressions} (or simply {\em tables}, as we are going to call them in this work) are used for the documentation of software programs (or, more generally, engineering designs). The concept was first introduced by David Parnas in the 1970s (cf. \cite{alspaugh1992software} and the references therein), and since then there has been a proliferation of different semantics and syntax variations developed. In this work we use a simple variation~\cite{wassyng2003tabular}, already used in the development of critical software. Figure \subref*{table} describes the structure of such a table. The {\em Conditions} column contains predicates $p_i(x)$, and the {\em Functions} column contains functions $f_i(x)$ (constant values are regarded as zero-arity functions). The table works as follows: If $p_i(x)=True$ for an input $x$, then $f_i(x)$ will be the output $T(x)$ of table $T$. To work properly, the predicates in $T$ must satisfy the {\em disjointness} and {\em completeness} properties: 
\begin{description}
\item[Completeness:] $p_1(x) \vee p_2(x) \vee ... \vee p_n(x) = True$ 
\item[Disjointness:] $p_i(x) \wedge p_j(x) = False, \forall i,j \in [n],$
\end{description}
i.e., for any input $x$ exactly one of the table predicates is $True$. 
Note that a table with a single row $(p_1(x),f_1(x))$ satisfies these properties only when $p_1(x)=True,\ \forall x$. This is important, because
we are going to work with an equivalent representation that has only {\em single-row} tables.

\begin{figure}[t]
\centering
  \subfloat[A table $T$ with input $x$ and output $T(x)$]{%
    {\noindent\begin{minipage}[t]{0.48\textwidth}
      $ x \rightarrow $ 
      \begin{tabular}{|c|c|}  
      \hline
      Condition & Function \\ \hline
      $p_1(x)$ & $f_1(x)$ \\ \hline
      $p_2(x)$ & $f_2(x)$ \\ \hline
      $\ldots$ & $\ldots$ \\ \hline
      $p_n(x)$ & $f_n(x)$ \\ \hline
      \end{tabular} 
      $ \rightarrow T(x)$ \linebreak \ 
      \end{minipage}} \label{table}} \hfill  
  \subfloat[A table graph]{%
    \includegraphics[width=.48\textwidth]{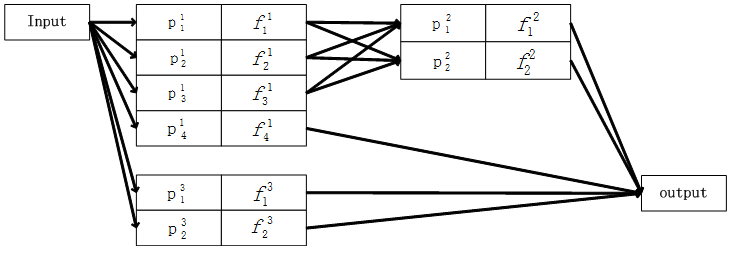} \label{tablegraph}} \\
  \subfloat[The transformed table graph $G$ of the graph in \protect\subref{tablegraph}]{%
    \includegraphics[width=.48\textwidth]{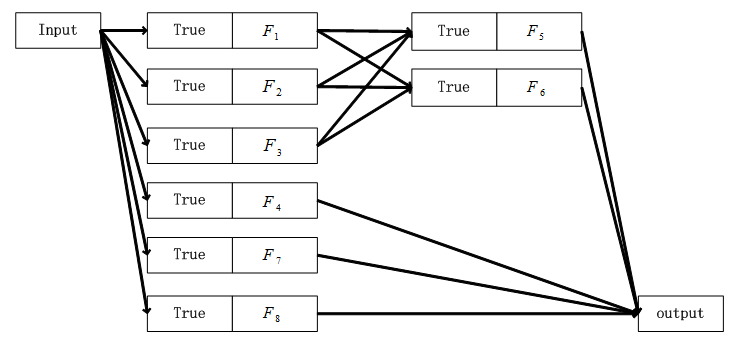} \label{ourtablegraph}}\hfill
  \subfloat[The structure graph $G_{struc}$ of table graph $G$ in \protect\subref{ourtablegraph}]{%
    \includegraphics[width=.48\textwidth]{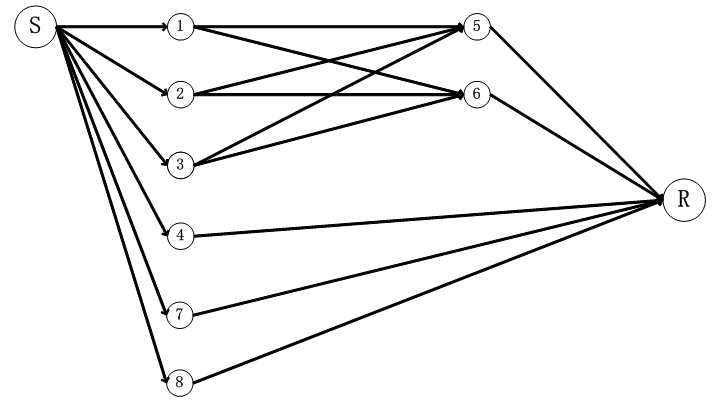}  \label{structuregraph}}
    \caption{Tables, table graphs, and their transformation.}\label{fig:1}
\end{figure}

A table can be used to represent a module of a program (or a design). Then, the whole program can be documented as a directed graph with its different tables being the nodes of the graph, a source node {\em Input} representing the inputs to the program, and a sink node {\em Output} representing the outputs of the program. Every table is connected to the tables or {\em Input} where it gets its input from, and connecting to the tables or {\em Output}, depending on where its output is forwarded for further use. An example
of such a graph can be seen in Figure \subref*{tablegraph}. We will make the following common assumption (achieved, for example, by loop urolling) for table graphs:
\begin{assum}   \label{assumdag}
The table graph of a program is acyclic.
\end{assum}

\subsubsection{Transformation to single-row tables}

Our algorithms will work with single row tables. Therefore, we show how to transform the general tables and graph in Figure \subref*{tablegraph} into an equivalent representation with one-row tables.  The transformation is done in steps:
\begin{itemize}
\item As a first step, a table with $n$ rows is broken
down into $n$ single-row tables. Note that, in general, the row of each multi-row table may receive different inputs
and produce different outputs (with different destinations) that the other table rows; we keep the inputs and outputs
of this row unchanged, when it becomes a separate table.    
\item We enhance each $x_j$ of the external input $(x_1,x_2,...,x_s)$ (introduced to the program at node {\em Inputs}, and transmitted through outgoing edges to the rest of the original table graph) with a special symbol $\top$, to produce 
a new input $(x'_1,...,x'_s)$, with $x'_j=(\top,x_j),\ j=1,\ldots,s$.
\item Let $(P(w),F(w))$ be the {\em (predicate, function)} pair for the new (single-row) table, corresponding to an old row $(p(x),f(x))$ of an original table. As noticed above, $P(w)=True$, no matter what the input $w=(w_1,w_2,\ldots, w_k)$ is. 
Let $x=(x_1,x_2,...,x_k)$, where $w_j=(\top,x_j)$  (i.e., the $x_j$'s are the $w_j$'s without the initial special symbol extension). 
The new function $F(w)$ is defined as follows:
\[ F(w) = \left\{ 
  \begin{array}{ll}
      (\top,f(x)), & \text{if } p(x)=1 \\
      (\bot,\bot), & \text{if } p(x)=0.
  \end{array} \right. \]
Note that $\top$ and $\bot$ are special symbols, not related to the boolean values $True$ and $False$. In particular,
from now on, we will assume that in case the first part of a table output is $\bot$, then this output will be recognized as a bogus
argument, and will not be used as an input by a subsequent table. The reason for adding $\top,\bot$ to $F$ 's output will become 
apparent in Section \ref{sec:encode}.
\end{itemize}

After this transformation, we will get a transformed graph (Figure \subref*{ourtablegraph} shows the transformed graph of
the example in Figure \subref*{tablegraph}). From now on, by {\em table graph} we mean a transformed graph, by
\emph{external input} we mean an input that comes directly from {\em Input}, and by \emph{external output} we mean an output that goes directly to {\em Output}. An \emph{intermediate input} will be an input to a table that is the output of another table, and an \emph{intermediate output} will be an output of a table that is used as an input to another table.

Since Assumption \ref{assumdag} holds, we can inductively order the tables of a table graph $G$ in {\em levels} as follows:
\begin{itemize}
\item An external input from node {\em Input} can be regarded as the output of a virtual level 0 table.
\item A table is {\em at level $k$} ($k>0$) if it has at least one incoming edge from a level $k-1$ table, 
and no incoming edges from tables of level larger than $k-1$.
\end{itemize}
As a consequence, we can traverse a table graph (or subgraph) in an order of increasing levels; we call such a traversal {\em consistent}, because it allows for the correct evaluation of the graph tables (when evaluation is possible), starting from level $0$.
The graph structure of a table graph $G$ can be abstracted by a corresponding {\em structure graph} $G_{struc}$, which
replaces each table in $G$ with a simple node. For example, Figure \subref*{structuregraph} shows the structure graph
for the table graph in Figure \subref*{ourtablegraph}. 

If $TS=\{T_1,T_2,\ldots,T_n\}$ is the set of tables in the representation of a program, we will use the shorthand 
\[ G=[TS,G_{struc}] \]
to denote that its transformed table graph $G$ contains both the table information of $TS$ and the graph structure $G_{struc}$.

\subsection{Cryptographic notation and definitions}

We introduce some basic cryptographic notation and definitions that we are going to use, following the notation in \cite{goldwasser2013reusable}.
We denote by $negl(K)$ a \emph{negligible function of K}, i.e., a function such that
for all $c > 0$, there exists $K_0$ with the property that $negl(K) < K^{-c},\ \forall K>K_0$. 
We use the notation {\em p.p.t.} as an abbreviation of {\em probabilistic polynomial-time} when referring to algorithms. 
Two ensembles ${\{X_{K}\}}_{K \in \mathbb{N}}$ and ${\{Y_{K}\}}_{K \in \mathbb{N}}$, where $X_{K},Y_{K}$ are probability distributions over $\{0,1\}^{l(K)}$ for a polynomial $l(\cdot)$, are \emph{computationally indistinguishable} iff for every p.p.t. algorithm $D$,  
\[ |Pr[D(X_K,1^K)=1]-Pr[D(Y_K,1^K)=1]| \leq negl(K).  \]
We are going to use deterministic private key encryption (symmetric key encryption) schemes, defined as follows (see, e.g., 
\cite{katz2014introduction}):

\begin{defin}[Private key encryption scheme]  \label{symmetrickey}
A private key encryption scheme $SE$ is a tuple of three polynomial-time algorithms $(SE.KeyGen, SE.Enc, SE.Dec)$
with the following properties:
\begin{enumerate}[(1)]
\item The {\em key generation algorithm} $SE.KeyGen$ is a probabilistic algorithm that takes the security parameter $K$ as input, 
and outputs a key $sk \in \{0,1\}^{K}$.
\item The {\em encryption algorithm} $SE.Enc$ is a deterministic algorithm that takes a key $sk$ and a plaintext $M \in
\{0,1\}^{m(K)},\ m(K)>K$ as input, and outputs a ciphertext $C=SE.Enc(sk,M)$.
\item The {\em decryption algorithm} $SE.Dec$ is a deterministic algorithm that takes as input a key $sk$ and a ciphertext $C$, 
and outputs the corresponding plaintext $M=SE.Dec(sk,C)$.
\end{enumerate}
\end{defin}

In order for an encryption to be considered secure, we require {\em message indistinguishability}, i.e., we require that an adversary 
must not be able to distinguish between two ciphertexts, even if it chooses the corresponding plaintexts itself. More formally:

\begin{defin}[Single message indistinguishability]     \label{messageindt}
A private key encryption scheme $SE = (SE.KeyGen, SE.Enc, SE.Dec)$ is \emph{single message indistinguishable} iff for any security parameter K, for any two messages $M,M' \in \{0,1\}^{m(K)},\ m(K)>K$, and for any p.p.t. adversary $A$,
\[ |Pr[A(1^K,SE.Enc(k,M))=1]-Pr[A(1^K,SE.Enc(k,M'))=1]| \leq negl(K), \]
where the probabilities are taken over the (random) key produced by $SE.KeyGen(1^K)$, and the coin tosses of $A$.
\end{defin}
An example of a private key encryption scheme satisfying Definitions \ref{symmetrickey} and \ref{messageindt} is the block cipher 
Data Encryption Standard (DES) in \cite{standard1977fips}.

\subsubsection{Fully Homomorphic Encryption (FHE)}   \label{FHEdef}
 
We present the well-known definition of a powerful cryptographic primitive that we will use heavily in our protocols, namely the \emph{Fully Homomorphic Encryption (FHE)} scheme (see \cite{vaikuntanathan2011computing} and \cite{goldwasser2013reusable}
for the history of these schemes and references).
This definition is built as the end-result of a sequence of definitions, which we give next:

\begin{defin}[Homomorphic Encryption]
A homomorphic (public-key) encryption scheme $HE$ is a tuple of four polynomial time algorithms $(HE.KeyGen, HE.Enc, HE.Dec, HE.Eval)$ with the following properties:
\begin{enumerate}[(1)]
\item $HE.KeyGen(1^K)$ is a probabilistic algorithm that takes as input a security parameter $K$ (in unary), 
and outputs a public key $hpk$ and a secret key $hsk$.
\item $HE.Enc(hpk,x)$ is a probabilistic algorithm that takes as input a public key $hpk$ and an input bit $x\in\{0,1\}$, 
and outputs a ciphertext $\phi$. 
\item $HE.Dec(hsk, \phi)$ is a deterministic algorithm that takes as input a secret key $hsk$ and a ciphertext $\phi$, 
and outputs a message bit.  
\item $HE.Eval(hpk,C,\phi_1,...,\phi_n)$ is a deterministic algorithm that takes as input the public key $hpk$, a circuit $C$ 
with $n$-bit input and a single-bit output, as well as $n$ ciphertexts $\phi_1,...,\phi_n$. It outputs a ciphertext $\phi_C$. 
\end{enumerate}
\end{defin}

We will require that $HE.Eval$ satisfies the {\bf compactness} property: For all security parameters $K$, there exists a polynomial $p(\cdot)$ such that for all input sizes $n$, for all $\phi_1,...,\phi_n$, and for all $C$, the output length of $HE.Eval$ is at most $p(n)$ bits long.

\begin{defin}[C-homomorphic scheme]  \label{homomorphism}
Let $C=\{C_n\}_{n \in \mathbb{N}}$ be a class of boolean circuits, where $C_n$ is a set of boolean circuits taking n bits as input. A scheme $HE$ is {\em C-homomorphic} iff for every polynomial $n(\cdot)$, sufficiently large $K$, circuit $C \in C_n$, 
and input bit sequence $x_1,...,x_n$, where $n=n(K)$, we have
\begin{equation} \label{eqfhe}
Pr\left[HE.Dec(hsk,\phi)   \neq  C(x_1,...,x_n), \text{ where }
  \begin{array}{l}
    (hpk,hsk)  \leftarrow  HE.KeyGen(1^K) \\
    \phi_i  \leftarrow  HE.Enc(hpk,x_i),\ i=1,\ldots,n \\
    \phi  \leftarrow  HE.Eval(hpk,C,\phi_1,\ldots,\phi_n)
  \end{array} 
  \right] \leq negl(K),
\end{equation}
where the probability is over the random bits of $HE.KeyGen$ and $HE.Enc$.
\end{defin}

\begin{defin}
A scheme $HE$ is {\em fully homomorphic} iff it is homomorphic for the class of all arithmetic circuits over $GF(2)$.
\end{defin}

\begin{defin}[IND-CPA security]  \label{defINDCPA}
A scheme HE is {\em IND-CPA secure} iff for any p.p.t. adversary A, 
\begin{multline*}
  |Pr[(hpk,hsk)\leftarrow HE.KeyGen(1^K):A(hpk,HE.Enc(hpk,0))=1]-\\ 
  Pr[(hpk,hsk)\leftarrow HE.HeyGen(1^K):A(hpk,HE.Enc(hpk,1))=1]| \leq negl(K).
\end{multline*}
\end{defin}

Fully homomorphic encryption has been a concept known for a long time, but it wasn't until recently that Gentry~\cite{gentry2009fully}  gave a feasible implementation of FHE. Our work is independent of particular FHE implementations; we only require their existence. 
For simplicity, sometimes we write $FHE.Enc(x)$ when the public key $hpk$ isn't needed explicitly. For an $m$-bit string $x=x_1...x_m$, we write $FHE.Enc(x)$ instead of the concatenation $FHE.Enc(x_1)\odot\ldots\odot FHE.Enc(x_m)$, and we do the same for $FHE.Dec$ as well. Similarly, for $FHE.Eval$ with a circuit $C$ as its input such that $C$ outputs $m$ bits $C_1,C_2,\ldots,C_m$, sometimes we write $FHE.Eval(hpk,C,FHE.Enc(hpk,x))$ to denote the concatenation
$FHE.Eval(hpk,C_1,FHE.Enc(hpk,x))\odot\ldots\odot FHE.Eval(hpk,C_m,FHE.Enc(hpk,x))$.

We usually use $\lambda=\lambda(K)$ to denote the ciphertext length of a one-bit FHE encryption.
 
Next, we present the {\em multi-hop homomorphism} definition of \cite{gentry2010hop}. Multi-hop homomorphism is an important property for our algorithms, because it allows using the output of one homomorphic evaluation as an input to another homomorphic evaluation. 

An ordered sequence of functions $\vec{f}=\{f_1,\ldots,f_t\}$ is \emph{compatible} if the output length of $f_j$ is the same as the input length of $f_{j+1}$, for all $j$. The composition of these functions is denoted by $(f_t\circ...\circ f_1)(x)=f_t(\ldots f_2(f_1(\cdot))...)$.
Given a procedure $Eval(\cdot)$, we can define an {\em extended procedure} $Eval^*$ as follows: $Eval^*$ takes as input the public key $pk$, a compatible sequence $\vec{f}=\{f_1,\ldots ,f_t\}$, and a ciphertext $c_0$. For $i=1,2,\ldots,t$, it sets $c_i \leftarrow Eval(pk,f_i,c_{i-1})$, outputting the last ciphertext $c_t$.

\begin{defin}[Multi-hop homomorphic encryption scheme]     \label{defmultihop}
Let $i=i(K)$ be a function of the security parameter $K$. A scheme $HE=(HE.KeyGen,HE.Enc,HE.Dec,HE.Eval)$ is an {\em $i$-hop homomorphic encryption scheme} iff for every compatible sequence $\vec{f}=\{f_1,\ldots,f_t\}$ with $t\leq i$ functions, every input $x$ to $f_1$, every (public,secret) key pair $(hpk,hsk)$ in the support of $HE.KeyGen$, and every $c$ in the support of $HE.Enc(hpk,x)$,
\[ HE.Dec(hsk,Eval^*(hpk,\vec{f},c))=(f_t\circ\ldots\circ f_1)(x). \]
$HE$ is a {\em multi-hop homomorphic encryption scheme} iff it is $i$-hop for any polynomial $i(\cdot)$.
\end{defin}

Not all homomorphic encryption schemes satisfy this property, but \cite{gentry2010hop}, \cite{vaikuntanathan2011computing} show that it holds for fully homomorphic encryption schemes. In our algorithms we will use FHE schemes, that also satisfy the IND-CPA security property of Definition \ref{defINDCPA}.

\subsubsection{Bit commitment protocols}   \label{bitcommit1}

Following Naor \cite{naor1991bit}, a Commitment to Many Bits (CMB) protocol is defined as follows: 

\begin{defin}[Commitment to Many Bits (CMB) Protocol] \label{bitcommit2}
A CMB protocol consists of two stages:
\begin{description}
\item[The commit stage:] Alice has a sequence of bits $D=b_1b_2...b_m$ to which she wishes to commit to Bob. She and Bob enter a message exchange, and at the end of the stage Bob has some information $Enc_D$ about $D$.
\item[The revealing stage:] Bob knows $D$ at the end of this stage.
\end{description}
The protocol must satisfy the following property for any p.p.t. Bob, for all polynomials $p(\cdot)$, and for a large enough security parameter $K$:
\begin{itemize}
\item For any two sequences $D=b_1,b_2,...,b_m$ and $D'=b'_1,b'_2,...,b'_m$ selected by Bob, following the commit stage Bob cannot guess whether Alice committed to $D$ or $D'$ with probability greater than $1/2+1/p(K)$.
\item Alice can reveal only one possible sequence of bits. If she tries to reveal a different sequence of bits, then she is caught with probability at least $1-1/p(K)$.
\end{itemize}
\end{defin}
We are going to use the construction of a CMB protocol by Naor \cite{naor1991bit}.

\subsection{Evaluation and verification}  \label{evaluate}

\noindent{\bf Evaluation: }Tabular expressions (tables) can be used to represent a program at different levels of abstraction. In this work, we concentrate 
on two levels: the {\em specifications} level, and the {\em implementation} level. Our goal is to design protocols that will
allow a third party (usually called the {\em verifier} in our work) to securely and truthfully verify the compliance of a design
at the implementation level (represented by a table graph $G$) with the publicly known specifications (represented by
a table graph $G^{spec}$), without publicly disclosing more than the structure graph $G_{struc}$ of $G$.
   
We are going to abuse notation, and use $G^{spec}$ and $G$ also as functions $G^{spec}: D^{spec} \rightarrow R^{spec}$, and
$G:D \rightarrow R$. Wlog we assume that these two functions are over the same domain and have the same range, i.e., $D^{spec}=D$ and $R^{spec}=R$.\footnote{In general, the two pairs ($D, R$) and ($D^{spec}, R^{spec}$) may not be the same. Nevertheless, in this case there must be a predetermined mapping provided by the implementation designer, which converts ($D_{spec}, R_{spec}$) to ($D, R$) and vice versa, since, otherwise, the specifications verification is meaningless as a process.}

The evaluation of $G^{spec}, G$ (or any table graph), can be broken down into the evaluation of the individual tables
level-by-level. The {\em evaluation} of a table $T_i$ at input $x^i=(x^i_1,x^i_2,\ldots,x^i_k)$ produces an output $y^i=T_i(x^i)$. Note that
some inputs $x^i_j$ may be {\em null}, if they come from non-evaluated tables, and if $T_i$ is not evaluated at all, 
then $T_i(x^i)=null$. In general, if some $x^i_j$ is $(\bot,\bot)$ or \emph{null}, then we set $T_i(x^i)=null$.

Given an external input $X$ to the table graph $G$, the \emph{evaluation} of $G(X)$, is done as follows:
Let $TSO=(T_{s_1},...,T_{s_n})$ be an ordering of the tables in $G$ in increasing level numbers (note that this is also a topological
ordering, and that the first table is the virtual table of level $0$). Let $T_{i_1},...,T_{i_s}$ be the tables that have external outputs.
Then the evaluation of $G$ is the evaluation of its tables in the $TSO$ order, and $G(X)$ is the set of the external outputs of $T_{i_1},...,T_{i_s}$. We call this process of evaluating $G(X)$ {\em consistent}, because we will never try to evaluate
a table with some input from a table that has not been already evaluated. All graph evaluations will be assumed to be consistent
(otherwise they cannot be deterministically defined and then verified).
 
\noindent{\bf Verification: }The verification processes we deal with here work in two phases: 
\begin{itemize}
\item During the first phase, the table graph $G$ is analyzed, and a set of (external or internal) inputs is generated, 
using a publicly known verification test-case generation algorithm $VGA$. 
\item In the second phase, $G$ and $G^{spec}$ are evaluated on the set of inputs generated during the first phase. If the evaluations coincide, we say that $G$ {\em passes the verification}.
\end{itemize}

In our setting, $VGA$ is any algorithm that takes the publicly known $G_{struc},G^{spec}$ as input, and outputs information that guides the verification (such as a set of external inputs to the table graph $G$ or certain paths of $G$). We emphasize that the output of $VGA$ is not necessarily external inputs. Our protocols can also work with any $VGA$ which generates enough information for the verifier to produce a corresponding set of external inputs, possibly by interacting with the designer of implementation $G$. We also allow for verification
on a (possibly empty) predetermined and publicly known set $CP$ of $(input,output)$ pairs, i.e., for every $(X,Y)\in CP$ with $X$ being a subset of the external inputs and $Y$ a subset of external outputs, $G$ passes the verification iff $G(X)$ can be evaluated and $G(X)=Y$. The pairs in $CP$ are called {\em critical points}.   

Our protocols will use a {\em security parameter} $K$, represented in unary as $K$ 1's, i.e., as $1^K$. Now, we are 
ready to define a {\em trusted verifier} (or {\em trusted verification algorithm}) (we use the two terms interchangeably): 
\begin{defin}[Trusted verifier] \label{verfier}
A {\em trusted verifier} $V$ is a p.p.t. algorithm that uses $r$ random bits, and such that
\[ Pr_r\left[V(1^K,G,G^{spec},CP,VGA)(r) =  \left\{
\begin{array}{ll}  
0, & \text{if } \exists X \in EI \text{: } G(X) \neq G^{spec}(X) \\
0, & \text{if } (CP\not=\emptyset) \wedge (\exists (X,Y) \in CP \text{: }G(X)\neq Y) \\ 
1, & \text{if the previous two conditions aren't true} 
\end{array}\right. \right] \geq 1-negl(K) \]
where $EI\subseteq D$ is a (possibly empty) set of external inputs generated by $V$ itself. 
\end{defin}
A verifier that satisfies Definition \ref{verfier} is called trusted, because it behaves in the way a verifier is supposed to behave whp:
Essentially, a trusted verifier uses the publicly known $K,G^{spec},CP,VGA$ and a publicly known $G$,
to produce a set of external inputs $EI$; it accepts iff the verification doesn't fail at a point of $EI$ or $CP$. During its running,
the verifier can interact with the designer of $G$. While $G$ is publicly known, the operations of the verifier are straight-forward,
even without any interaction with the designer. The problem in our case is that the verifier has only an {\em encryption}
of $G$, and yet, it needs to evaluate $G(X)$ in order to perform its test. This paper shows how to succeed in doing that,
interacting with the designer of $G$, and without leaking more information about $G$ than the publicly known $G_{struc}$.  

In what follows, a verifier will be a part of a {\em verification scheme}, generally defined as follows:
\begin{defin}[Verification scheme]  \label{defe}
A {\em verification scheme} $VS$ is a tuple of p.p.t. algorithms $(VS.Encrypt, VS.Encode, VS.Eval)$ such that
\begin{itemize}
\item $VS.Encrypt(1^K, G)$ is a p.p.t. algorithm that takes a security parameter $1^K$ and a table graph $G$ as input, 
and outputs an encrypted table graph $G'$. 
\item $VS.Encode$ is a p.p.t. algorithm that takes an input $x$ and returns an encoding $Enc_{x}$.	
\item $VS.Eval$ is a p.p.t. algorithm that takes a security parameter $K$ and a public $Certificate$ as input, has an honest verifier $V$ satisfying Definition \ref{verfier} hardcoded in it, and outputs $1$ if the verification has been done correctly (and $0$ otherwise).   
\end{itemize}
\end{defin}

From now on, it will be assumed (by the verifier and the general public) that the encryption by the designer is done
(or {\em looks like} it has been done) by using a publicly known algorithm $VS.Encrypt$ (with a secret key, of course). If the 
designer's encryptions don't comply with the format of $VS.Encrypt$'s output, the verification fails automatically.

In general, it may be the case that an algorithm claiming to be a trusted verifier within a verification scheme, doesn't satisfy Definition \ref{verfier} (either 
maliciously or unintentionally). Such a malicious verifier may claim that a design passes or doesn't pass the verification process,
when the opposite is true. In order to guard against such a behaviour, we will require that there is a piece of public information, that
will act as a {\em certificate}, and will allow the {\em exact} replication of the verification process by any other verifier and at any
(possibly later) time; if this other verifier is a known {\em trusted} verifier, it will detect an incorrect verification with high probability (whp). 
We emphasize that the interaction (i.e., the messages exchange) shown in Figure \ref{protocol1} is done {\em publicly}.
In fact, we will use the transaction record as a {\em certificate} (cf. Definition \ref{defe}), that can be used to replicate 
and check the verification, as described above.

Hence, we will require that our verification schemes are:
\begin{itemize}
\item {\em secure}, i.e., they don't leak the designer's private information, and
\item {\em correct}, i.e., they produce the correct verification result whp.
\end{itemize}
More formally, we define:

\begin{defin}[Correctness]   \label{correctness3}
A verification scheme is {\em correct} iff the following holds:  
$VS.Eval(1^K,Certificate)=1$ if and only if both of the following hold:
 \begin{eqnarray}  
 Pr_r[V(1^K,G,G^{spec},VGA_r,CP)(r)=V(1^K,G',G^{spec},VGA_r,CP)(r)] \geq 1-negl(K) & \label{eqcorrectness1} \\
 Pr_r[V(1^K,G',G^{spec},VGA_r,CP)(r)=V'(1^K,G',G^{spec},VGA_r,CP)(r)] \geq 1-negl(K) & \label{eqcorrectness2}
 \end{eqnarray}
where $V$ is the honest verifier hardwired in $VS.Eval$, and $G'$ is the table graph produced by $VS.Encrypt$.
\end{defin}
Condition \eqref{eqcorrectness2} forces $V$ to be trusted, and condition \eqref{eqcorrectness1} ensures that its verification
is correct whp, even when applied to $G'$ instead of $G$; together ensure the correctness of the verification whp.

Just as it is done in~\cite{goldwasser2013reusable}, we give a standard definition of security:
\begin{defin}[Security]   \label{security3}
For any two pairs of p.p.t. algorithms $A=(A_1,A_2)$ and $S=(S_1,S_2)$, consider the two experiments in Figure \ref{exprl}.
\begin{figure}[h!]
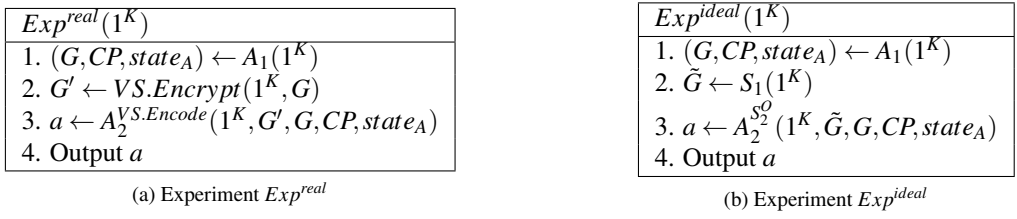

\centering
  \subfloat[Experiment $Exp^{real}$]{%
    {\noindent\begin{minipage}[t]{0.48\textwidth}
    \centering
    \begin{tabular}{|l|}
    \hline
    $Exp^{real}(1^K)$ \\ \hline
    1. $(G,CP,state_A) \leftarrow A_1(1^K)$ \\
    2. $G' \leftarrow VS.Encrypt(1^K,G)$  \\
    3. $a \leftarrow A_2^{VS.Encode}(1^K,G',G,CP,state_A)$  \\
    4. Output $a$  \\ \hline
    \end{tabular}
    \end{minipage}}} \hfill
  \subfloat[Experiment $Exp^{ideal}$]{%
    {\noindent\begin{minipage}[t]{0.48\textwidth}
    \centering
    \begin{tabular}{|l|}
    \hline
    $Exp^{ideal}(1^K)$ \\ \hline
    1. $(G,CP,state_A) \leftarrow A_1(1^K)$  \\
    2. $\tilde{G} \leftarrow S_1(1^K)$  \\
    3. $a \leftarrow A_2^{S^{O}_2}(1^K,\tilde{G},G,CP,state_A)$  \\
    4. Output $a$  \\ \hline
    \end{tabular}
    \end{minipage}}}   
    \caption{Experiments $Exp^{real}$ and $Exp^{ideal}$}  \label{exprl}
\end{figure}

A verification scheme $VS$ is \emph{secure} iff there exist a pair of p.p.t. algorithms (simulators) $S=(S_1,S_2)$, 
and an oracle $O$ such that for all pairs of p.p.t. adversaries $A=(A_1,A_2)$, the following is true for any p.p.t. algorithm $D$:
\[ \left|Pr[D({Exp^{ideal}(1^K)},1^K)=1]-Pr[D({Exp^{real}(1^K)},1^K)=1]\right| \leq negl(K),\]
i.e., the two experiments are computationally indistinguishable.
\end{defin}

\begin{defin}[Secure and trusted verification scheme]   \label{defenc} 
A verification scheme of Definition \ref{defe} is {\em secure and trusted} iff it satisfies Definitions \ref{correctness3} and \ref{security3}.
\end{defin}
 
\begin{figure}[h!]
\begin{center}
\includegraphics[width=\textwidth]{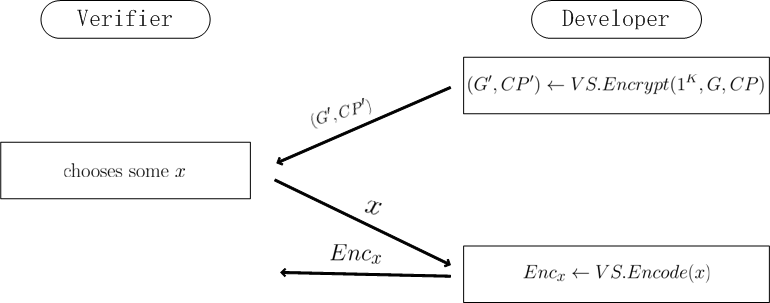}
\end{center}
\caption{The initial steps of the generic $VS$ protocol in Definition \ref{defe}.}
\label{protocol1}
\end{figure}

\begin{rem}  \label{remnotoracle}
$VS.Encode$ and $S^{O}_2$ in Step 3 of the two experiments in Figure \ref{exprl} are {\em not} oracles 
used by $A_2$. In these experiments, $A_2$ plays the role of a (potentially malicious) verifier and interacts with the developer as shown in Figure \ref{protocol1}. More specifically, in $Exp^{real}$, $A_2$ picks an input $x$, {\em asks the developer to run $VS.Encode(x)$} (instead of querying an oracle), and receives the answer. In $Exp^{ideal}$, $A_2$ again asks the developer, but, unlike $Exp^{real}$, the latter runs $S^{O}_2$ instead of $VS.Encode$ and provides $A_2$ with the answer. Hence, whenever we say that $A_2$ queries $VS.Encode$ or $S^{O}_2$, we mean that $A_2$ asks the developer to run $VS.Encode$ or $S^{O}_2$ respectively, and provide the answer. Note that $O$ {\em is} an oracle queried by $S_2$.
\end{rem}

In what follows, first we deal (in Section \ref{sec:honest}) with an {\em honest} designer/developer, i.e., a developer that answers honestly the verifier's
queries:  
\begin{defin}   \label{honest}
An \emph{honest} developer is a developer that calls $VS.Encode(x)$ to generate $Enc_x$ in Figure \ref{protocol1},
when queried by the verifier.
\end{defin} 
In the case of using the verifier to test a design, the developer obviously has no reason to not being honest. But in the case
of verifying the correctness of the design, the developer may not follow Definition \ref{honest}, in order to pass the verification
process (cf. Section \ref{sec:maldev}).


\section{Secure and trusted verification for honest developers}  \label{sec:honest}

\subsection{Construction outline}  \label{sec:construct}

In this section, we construct a secure verification scheme satisfying Definition \ref{defenc}. In a nutshell, we are looking for
a scheme that, given an encrypted table graph and its structure graph, will verify the {\em correctness} of evaluation
on a set of test external inputs, in a {\em secure} way; to ensure the latter, we will require that the intermediate table
inputs/outputs produced during the verification process are also encrypted.

\subsubsection{VS.Encrypt}

In order to encrypt the rhs functions of the tables as well as the intermediate inputs/outputs, we use {\em universal circuits} 
\cite{valiant1976universal}. If we represent each function $F_i$ in (transformed) table $T_i=(True,F_i)$ as a boolean circuit
$C_i$ (see Appendix \ref{sec:circuits} for methods to do this), then we construct a universal circuit $U$ and a string $S_{C_i}$ for each $C_i$, so that for any input $x$, $U(S_C,x)=C(x)$.
Following \cite{sander1999non} (a method also used in \cite{goldwasser2013reusable}), $S_{C_i}$ and $x$ can be encrypted,
while the computation can still be performed and output an encryption of $C(x)$.  
 
Therefore, $VS.Encrypt$ fully homomorphically encrypts $S_{C_i}$ with FHE's public key $hpk$ to get $E_{C_i}$, and replaces each table $T_i$ with its encrypted version $T'_i=(True,E_{C_i})$. Then, if a verifier wants to evaluate $T'_i$ at $x$,
it gets the FHE encryption $x'$ of $x$, and runs $FHE.Eval(hpk,U,E_{C_i},x')$ to get $T'_i(x')=FHE.Eval(hpk,U,E_{C_i},x')$. Because $FHE.Eval(hpk,U,E_{C_i},x')=FHE.Enc(hpk,C(x))$ holds, we have that $T'_i(x')=FHE.Enc(hpk,C(x))$. Note that
the encrypted graph $G'$ that $VS.Encrypt$ outputs maintains the same structure graph $G_{struc}$ as the original table 
graph $G$, and that $VS.Encrypt$ outputs $hpk,U$ in addition to $G'$. Algorithm \ref{encrypt} implements $VS.Encrypt$.

\begin{algorithm}[h!]
\caption{$VS.Encrypt(1^K,G)$} 
\begin{algorithmic}[1]
\State{$(hpk,hsk) \leftarrow FHE.KeyGen(1^K)$ }
\State {Construct a universal circuit $U$ such that, for any circuit $C$ of size $s$ and depth $d$, a corresponding string $S_C$  can be efficiently (in terms of $s$ and $d$) computed from $C$, so that $U(S_C,x)=C(x)$ for any input $x$.}
\State Suppose $C$ outputs $m$ bits. Construct $m$ circuit $U_1,...,U_m$ such that for input $x$ and any $i\in [m]$, $U(x,S_C)$ outputs the $i$th bit of $U(x,S_C)$. \label{encryptm}
\ForAll {$T_i \in TS, i \in \{1,...,n\}$}   \label{lineencrypt1}
\State Let $C_i$ be the circuit that $C_i(x)=F_i(x)$
\State Construct the string $S_{C_i}$ from $C_i$
\State $E_{C_i} \leftarrow FHE.Enc(hpk,S_{C_i})$
\State $T'_i \leftarrow (True,E_{C_i})$
\EndFor   \label{lineencrypt2}
\State $TS' \leftarrow \{T'_1,...,T'_n\}$
\State $G'_{struc} \leftarrow G_{struc}$
\State \Return {$G'=[TS', G'_{struc}], hpk, U=(U_1,...,U_m)$} 
\end{algorithmic}    \label{encrypt}
\end{algorithm}

\subsubsection{VS.Encode}   \label{sec:encode}

$VS.Encode$ is going to address two cases:
\begin{enumerate}
\item Suppose the verifier is evaluating an encrypted table $T'_i$ whose output is an external output. From the construction of $T'_i$, we know that its output is an FHE-encrypted ciphertext. But the verifier needs the {\em plaintext} of this ciphertext in order for verification to work. We certainly cannot allow the verifier itself to decrypt this ciphertext, because then the verifier (knowing the secret key of the encryption) would be able to decrypt the encrypted circuit inside $T'_i$ as well. What we can do is to allow the verifier to ask the developer for the plaintext; then the latter calls $VS.Encode$ to decrypt this ciphertext for the verifier.
\item Suppose the verifier is evaluating an encrypted table $T'_i$ whose output is an internal output used as an input to another table. In this case, we cannot allow the verifier to simply ask the developer to decrypt this output as before; that would give away the intermediate outputs, which are supposed to be kept secret. At the same time, the verifier must be able to figure out whether the actual (unencrypted) value of this intermediate output is $(\bot,\bot)$ or not, i.e., ``meaningful" or not. More specifically, $V$ should be able to tell whether intermediate output $FHE.Enc(\top,b)$ contains $\top$ or $\bot$. (Recall from Section \ref{sec:prelim} that an original table $T_i$ in $G$ outputs a symbol $\bot$ if the predicate in the initial table graph is not satisfied, i.e., $T_i$'s output is not ``meaningful".) 

In Case 2, particular care must be taken when the verifier chooses an external input $x=(\top,a)$ for table $T'_i$. The verifier is required to use the FHE encryption $x'$ of $x$ in order to evaluate $T'_i$. An obvious way  to do this is to let the verifier compute $FHE.Enc(hpk,(\top,a))$, and evaluate 
\[ T'_i(FHE.Enc(hpk,(\top,a))) \leftarrow FHE.Eval(hpk,U,E_{C_i},FHE.Enc(hpk,(\top,a))). \]
However, this simple solution can be attacked by a malicious verifier as follows: The verifier chooses one of the intermediate outputs, and claims that this intermediate output is the encryption of the external input. Then Case 1 applies, and the verifier asks the developer for the output of $VS.Encode$. Then it is obvious that through this interaction with the developer, the verifier can extract some partial information about the intermediate output. We use $VS.Encode$ to prevent this malicious behaviour: For any external input chosen by the verifier, the latter cannot fully homomorphically encrypt the input by itself; instead, it must send the input to the developer, who, in turn, generates the FHE encryption by calling $VS.Encode$. 
\end{enumerate}
In order to allow $VS.Encode$ to distinguish between the two cases, we introduce an extra input parameter that takes a value from the set
of special 'symbols' $\{q_1,q_2\}$ meaning the following:
\begin{description}
\item[$q_1$:] $VS.Encode$ is called with $(i,x^i,null,q_1)$. Index $i$ indicates the $i$th table $T'_i \in G'$ and $x^i$ is an external input to $T'_i$ (Case 1). 
\item[$q_2$:] $VS.Encode$ is called with $(i,x^i,T'_i(x^i),q_2)$. Index $i$ indicates the $i$th table $T'_i \in G'$ and $x^i$ is an intermediate input to $T'_i$ (Case 2).
\end{description}
Also, we allow $VS.Encode$ to store data in a memory $M$, which is wiped clean only at the beginning of the protocol. Algorithm
\ref{encode} implements $VS.Encode$.

\begin{algorithm}[h!]
\caption{$VS.Encode(i,u,v,q)$}
\begin{algorithmic}[1]
\If { $q=q_1$}
\Comment Case 1
\If {$|u| \neq m$ or first component of $u$ is not $\top$} 
\Comment $m$ defined in line \ref{encryptm} of Algorithm \ref{encrypt}
\State \Return $null$
\Else
\State $w \leftarrow FHE.Enc(hpk,u)$
\label{basecase}
\State store $(i,w,null)$ in $M$
\EndIf
\State \Return $w$
\EndIf   \label{lineqkindencode}
\If {$q=q_2$}
\Comment Case 2
\ForAll  {$x^i_j \in u$}     
\Comment $u=(x^i_1,x^i_2,\ldots)$ 
\If {$x^i_j$ is an output of some $T'_k$, according to $G'$}
  \If{$\nexists (k,x^k, T'_k(x^k)) \in M$ such that $x^i_j = T'_k(x^k)$ }    \label{lineerror1}
  \Return $null$
  \ElsIf {$\exists (k,x^k,T'_k(x^k)) \in M$ such that $x^i_j = T'_k(x^k)$ and $T'_k(x^k)$ is an FHE encryption of $(\bot,\bot)$}   \label{lineerror2}
  \Return $null$
\EndIf
\ElsIf{$x^i_j$ is an external input, according to $G'$}
\If {$\nexists (i,w,null) \in M$ such that $x^i_j = w$ }    
\label{lineerror3}
\Return $null$
\EndIf
\EndIf
\EndFor
\ForAll {$k \in \{1,...,m\}$}
\State $s_k=FHE.Eval(hpk,U_k,u,E_{C_i})$    
\Comment recall that $U=(U_1,U_2,\ldots,U_m)$ 
\EndFor
\If {$[s_1s_2...s_m] \neq v$}
\label{lineerror4}
\Return $null$
\Else
\State store $(i,u,v)$ in $M$   \label{lineencode1}
\ForAll {$i \in \{1,...,m\}$}
\State $b_i \leftarrow FHE.Dec(hsk,s_i)$ \label{line30}
\EndFor
\If {$v$ is not an external output and $b_1b_2...b_{m/2} \neq \bot$}
\Return $\top$
\Else
\ \Return $b_{m/2+1}b_2...b_m$
\EndIf\label{lineencode2}
\EndIf
\EndIf
\end{algorithmic}   \label{encode}
\end{algorithm}

\subsubsection{VS.Eval}   \label{sec:vseval}

$VS.Eval$ is an algorithm that allows anyone to check whether the verification was done correctly. In order to achieve this, the interaction between the verifier and the developer is recorded in a public file $QA_E$. By reading $QA_E$, $VS.Eval$ can infer which inputs and tables the verifier evaluates, and what outputs it gets. This allows $VS.Eval$ to evaluate the tables on the same inputs, {\em using its own verifier}, and {\em at any time}, and check whether the outputs are the recorded ones. Obviously, by using only an honest verifier and $QA_E$, $VS.Eval$ essentially checks whether the verifier that interacted with the developer originally was also honest. 

More specifically, $VS.Eval$ takes input $(1^K, QA_E)$, and outputs $1$ or $0$ (i.e., accept or reject the verification, respectively). The record file $QA_E=\{(Q_1,A_1),...,(Q_n,A_n), VGA, s, CP, G', hpk, U\}$ records a sequence of {\em (verifier query, developer reply)} pairs $(Q_i,A_i)$, where $Q_i$ is the verifier query, and $A_i$ is the reply generated by the developer when it runs $VS.Encode(Q_i)$ (the last pair obviously records the verifier's output). It also records test-generator $VGA$, the set of critical points $CP$ used, and the random seed $s$ used by $VGA$ to produce the test points. Finally, it records the encrypted table
graph $G'$, the public FHE key $hpk$, and the universal circuit $U$. As mentioned above, the pairs $(Q_i,A_i)$ are public
knowledge. Algorithm \ref{eval} implements $VS.Eval$. $V'$ is the hardwired honest verifier.

\begin{algorithm}[h!]
\caption{$VS.Eval(1^K,QA_E$)}
\begin{algorithmic}[1]
\ForAll{ $(Q_i,A_i) \in QA_E$}
\State Suppose $Q_i$ is $(r,u,v,q_2)$
\State $T'_r(u) \leftarrow FHE.Eval(hpk,U,E_{C_r},u)$ 
\Comment recall that $T'_r=(True,E_{C_r})$ is in $G'$  \label{line3vseval}
\If {$T'_r(u) \neq v$}
\State \Return 0
\EndIf
\EndFor
\State Run honest verifier $V'(1^K,G',G^{spec},VGA(s),CP)$  \label{lineVGA}
\ForAll {$T'_i(x^i)$ that $V'$ chooses to evaluate}
\If{$\nexists (Q_j,A_j) \in QA_E$ with $Q_j=(i,x^i,T'_i(x^i),q_2)$}
\State \Return 0
\EndIf
\EndFor
\If {$V'$'s output $\not=$ $V$'s output}
\State \Return 0
\EndIf
\State \Return 1 \label{lineQA}
\end{algorithmic}   \label{eval}
\end{algorithm}

 \subsubsection{VS.Path}

We can provide the verifier with the ability to choose and evaluate paths in $G_{struc}$. The verifier picks a path in the table graph, passes it to the developer, and the latter runs $VS.Path$ in order to generate an external input that, when used, will lead to the evaluation of the tables on the path chosen by the verifier. Algorithm \ref{eval} implements $VS.Path$.

\begin{algorithm}[h!]
\caption{$VS.Path(T_1,...,T_p)$}    \label{Path}
\begin{algorithmic}[1]
\If {$T_1,\ldots,T_p$ form a path $P=T_1 \rightarrow T_2 \rightarrow \ldots \rightarrow T_p$ in $G_{struc}$}
\State Generate external input $X$ to the table graph $G$, so that the evaluation of $G(X)$ includes the evaluation of tables $T_1,\ldots,T_p$.
\State \Return $X$
\Else
\ \Return $null$
\EndIf
\end{algorithmic}   \label{vspath}
\end{algorithm}

\subsubsection{Encrypted table graph evaluation}

The evaluation of an encrypted table graph $G'$ is more complex than the evaluation of an unencrypted table graph $G$. The evaluation of $G'$ at an external input $X$ (i.e., $G'(X)$) is done by Algorithm \ref{evalalg}. 

\begin{algorithm}[h!]
\caption{$G'(X)$}
\begin{algorithmic}[1]
\State $V$ uses $X$ to set up external inputs for the tables
\ForAll {table $T'_i$ chosen in a consistent order} 
\ForAll { $x^{i}_j \in x^{i}$}
\Comment $x^i=(x^i_1,x^i_2,\ldots)$
\If {$x^{i}_j$ is an external input}
\If {$x^i_j\not= null$}
\State $V$ asks the developer to call $VS.Encode(i,x^{i}_j,null,q_1)$
\State $x^{i}_j \leftarrow VS.Encode(i,x^{i}_j,null,q_1)$ 
\Comment $x^i$ is being replaced with its encoding 
\Else
\State $T'_{i}(x^i)=null$
\label{geval1}
\State \textbf{break}
\EndIf
\Else
\State $T'_p=$ the table whose output is $x^i_j$ 
\If {$T'_p$'s output is $null$}
\State $T'_{i}(x^i) = null$    \label{geval2}
\State \textbf{break}
\ElsIf {$VS.Encode(p,x^p,T'_p(x^p))=\bot$ }
\State $T'_{i}(x^i) = null$
\State \textbf{break}
\Else \ $x^{i}_j \leftarrow T'_p(x^p)$
\Comment $T'_p(x^p)$ is already encoded
\EndIf
\EndIf
\EndFor
\If {$T'_{i}(x^i)=null$} \textbf{continue}
\Else 
\State $T'_i(x^i) \leftarrow FHE.Eval(hpk,U,E_{C_i},x^i)$
\Comment $V$ evaluates $T'_{i}(x^{i})$
\label{evalpt}
\State $V$ asks the developer to call $VS.Encode(i,x^{i},T'_{i}(x^{i}),q_2)$   \label{encodept}
\State $V$ receives $VS.Encode(i,x^{i},T'_{i}(x^{i}),q_2)$
\EndIf
\EndFor
\State $Y=\{y_{i_1},...,y_{i_s}\}$ are the external output values
\State \Return $Y$
\end{algorithmic}   \label{evalalg}
\end{algorithm}

\subsection{Correctness and security}     \label{sectionproof3}

We have described an implementation of $VS.Encrypt$, $VS.Encode$, $VS.Eval$ of Definition \ref{defe}. 
Note that $QA_E$ plays the role of $Certificate$ for $VS.Eval$. In this section we prove the compliance of our scheme
with Definitions \ref{correctness3} and \ref{security3}. Recall that $FHE$ is the fully homomorphic encryption scheme introduced in Section \ref{FHEdef}.

\subsubsection{Correctness}

\begin{theorem}    \label{thm3}
The verification scheme VS introduced in Section \ref{sec:construct} satisfies Definition \ref{correctness3}.
\end{theorem}
\begin{proof}
We will need the following lemma: 

\begin{lemma}   \label{lcorrectness}
When $VS.Eval$ outputs $1$, for any table $T' \in TS'$ and $x$, the output $y$ claimed by $V$ as $T'(x)$ must be equal to $T'(x)$, i.e., $V$ evaluates every table correctly.
\end{lemma}
\begin{proof}
Given that the verifier $V'$ used by $VS.Eval$ is an honest one, the lemma is obviously true.
\end{proof}

For the the {\em (query,answer)} pairs in $QA_E$ that do not belong to a consistent traversal of $G'$, Algorithm \ref{encode} (lines \ref{lineerror1}, \ref{lineerror2}, \ref{lineerror3}, \ref{lineerror4}) will return $null$ as the corresponding input encodings, and Algorithm \ref{evalalg} (lines \ref{geval1}, \ref{geval2}) will compute $null$ for tables with such inputs. Therefore, security and correctness are not an issue in this case. From now, on we will concentrate on the consistent table traversals and their inputs/outputs.

\begin{lemma}   \label{induction}
Given the external input set $X$, common to both $G'$ and $G$, for any table $T_i \in TS$ with input $x^i$ and output $T_i(x^i)$, there is a corresponding table $T'_i \in TS'$ and input $w^i$, such that
\begin{equation}  \label{eqinduction}
\begin{split}
w^i &= FHE.Enc(hpk,x^i) \\
T'_i(w^i) &= FHE.Enc(hpk,T_i(x^i)),
\end{split}
\end{equation}
where $hpk$ is the public key generated by $VS.Encrypt(1^K,G)$.
\end{lemma}
\begin{proof}
We prove this lemma by induction on the level $k$ of a table.

\noindent{\bf Base case:} Level $1$ tables have only external inputs, and, therefore, Algorithm \ref{encode} (line \ref{basecase}) returns $w^i=FHE.Enc(hpk,x^i)$. Also, from line \ref{evalpt} of Algorithm \ref{evalalg} we have 
 \begin{align*}
 T'_i(w^i) &= FHE.Eval(hpk,U,E_{C_i},w^i) \\
  &= FHE.Eval(hpk,U,E_{C_i},FHE.Enc(hpk,x^i))\\
  &= FHE.Enc(hpk,FHE.Dec(hsk,FHE.Eval(hpk,U,E_{C_i},FHE.Enc(hpk,x^i)))) \\
  &= FHE.Enc(hpk,U(S_{C_i},x^i)) = FHE.Enc(hpk,T_i)
 \end{align*} 
 
\noindent{\bf Inductive step:} Suppose that for any table $T_i \in TS$ whose level is smaller than $k+1$ and its input $x^i$, table $T'_i \in TS'$ and its input $w^i$ satisfy \eqref{eqinduction}.
Then we show that for any level $k+1$ table $T_j \in TS$ with input $x^j=\{x^j_1,x^j_2,...,x^j_s\}$, the level $k+1$ table $T'_j \in TS'$ and its \emph{encoded} input $w^j$ satisfy \eqref{eqinduction}.

For any $x^j_p \in x^j$, either $x^j_p=T_{i_p}(x^{i_p})$ for some $i_p$, or $x^j_p$ is an external input. In case $x^j_p$ is external, then $w^j_p$ satisfies \eqref{eqinduction} like in the base case. Otherwise, $T_{i_p}$ is a table of level smaller than $k+1$. Then, because of the inductive hypothesis,  
\begin{equation*}
w^j = \{w^j_1,...,w^j_s\} = \{FHE.Enc(hpk,x^j_1),...,FHE.Enc(hpk,x^j_s)\} =  FHE.Enc(hpk,x^j).
\end{equation*} 
(The last equation is valid because $FHE.Enc$ encrypts a string bit by bit.) 
  
Now, the second part of \eqref{eqinduction} is proven similarly to the base case.
\end{proof}

\begin{lemma}\label{linput}
The input-output functionality of $G'$ is the same as the input-output functionality of $G$.
\end{lemma}
\begin{proof}
Given the common external input set $X$ to both $G'$ and $G$, suppose $T_{i_1},...,T_{i_s} \in TS$ are the output level tables that are actually evaluated, and $y_1,\ldots,y_s$ their corresponding outputs. Then Lemma \ref{induction} implies that
$T'_{i_1}(w^{i_1})=FHE.Enc(hpk,y_1),\ldots,T'_{i_s}(w^{i_s})=FHE.Enc(hpk,y_s).$
Accordingly, for every $j \in [s]$, by asking the developer to call $VS.Encode(i_j,w^{i_j},T'_{i_j}(w^{i_j}),q_2)$, verifier $V$ gets $y_1,...,y_s$ as the output of Algorithm \ref{encode} (line \ref{line30}).  
\end{proof}  

Lemma \ref{induction} (and hence Lemma \ref{linput}) are based on the verifier evaluating every table correctly and its subgraph traversal being consistent. Assuming that the traversal is indeed consistent (something easy to check on $G_{struc}$), by running $VS.Eval$, we can know whether $V$ satisfies Lemma \ref{lcorrectness} ($VS.Eval$ outputs $1$). If this is the case, we know that $V$'s evaluation of $G'$ has the same result as the evaluation of $G$ with the same inputs, i.e., \eqref{eqcorrectness1} in Definition \ref{correctness3} holds. Moreover, if $VS.Eval(1^K,QA_E)$ outputs $1$, \eqref{eqcorrectness2} in Definition \ref{correctness3} also holds (see lines \ref{lineVGA} - \ref{lineQA} in Algorithm \ref{eval}).
\end{proof}

\subsubsection{Security}

Following the approach of~\cite{goldwasser2013reusable}, we prove the following
\begin{theorem}\label{thm1}
The verification scheme VS introduced in Section \ref{sec:construct} satisfies Definition \ref{security3}.
\end{theorem}
\begin{proof}
In Definition \ref{security3} the two simulators $S_1$ and $S_2$ simulate $VS.Encrypt$ and $VS.Encode$ respectively. In our implementation we have added $VS.Path$, so we add $S_3$ to simulate $VS.Path$, together with its own oracle $O_2$. We can think of $(S_2,S_3)$ as a combination that plays the role of $S$, and $(O_1,O_2)$ as a combination that plays the role of $O$ in $Exp^{ideal}$ of Figure \ref{exprl}: depending on the kind of query posed by $A_2$, $S_2^{O_1}$ replies if it is an "Encode" query, and $S_3^{O_2}$ replies if it is a "Path" query. The tuple of simulators $(S_1,S_2,S_3)$ which satisfies Definition \ref{defenc} is constructed as follows:
\begin{itemize}
\item $S_1(1^K,s,d,G_{struc})$ runs in two steps:
\begin{description}
\item[Step 1:] $S_1$ generates its own table graph $\tilde{G}=(\tilde{TS},G_{struc})$ with its own tables $\tilde{T}_i \in \tilde{TS}$. 
\item[Step 2:] $S_1$ runs $VS.Encrypt(1^K,\tilde{G})$ (Algorithm \ref{encrypt}) and obtains $G'',hpk,U$.
\end{description} 
\item $S_2$ receives queries from the $A_2$ of Definition \ref{security3}, and can query its oracle $O_1$ (described in Algorithm \ref{oracle}). $S_2^{O_1}$ actually simulates $VS.Encode$, passing the queries from $A_2$ to its $VS.Encode$-like oracle $O_1$. $O_1$ has a state (or memory) $[M]$ which is initially empty, and contains the mapping of the encrypted table output to the real table output for a given query. $S_2$ returns the outputs of $O_1$ as the answers to the queries of $A_2$. 
\item $S_3^{O_2}$ simulates $VS.Path$. It receives queries from $A_2$, and queries its oracle $O_2$ (described in Algorithm \ref{oracle2}) in turn. Then $S_2$ returns the output of $O_2$ to $A_2$. 
\end{itemize}   

\begin{algorithm}[h!]
\caption{$O_1[M](i,u,v,q)$}\
\begin{algorithmic}[1]   \label{sim2}
\If { $q=q_1$}
\Comment Case $q_1$
\If {$|u| \neq m$ or first component of $u$ is not $\top$}  
\Comment $m$ defined in line \ref{encryptm} of Algorithm \ref{encrypt}
\State \Return $null$
\Else
\State $p \leftarrow FHE.Enc(hpk,u)$
\State store $(i,(p,null),(u,null))$ in $M$
\State \Return $p$
\EndIf
\EndIf    \label{lineqkindo}
\If {$q=q_2$}    \label{lineofunc1}
\Comment Case $q_2$
\ForAll {$x^i_j \in u$}
\Comment $u=(x^i_1,x^i_2,\ldots)$ 
\If {$x^i_j$ is an output of some $T''_k$, according to $G''$}    \label{lineoerror3}
  \If{$\exists (k,(y^k,T''_k(y^k)),(z^k,T(z^k))) \in M$ such that $x^i_j=T''_k(y^k)$}
  \State $e^i_j \leftarrow T_k(z^k)$
  \If{$T_k(z^k) = (\bot,\bot)$} 
  \Return $null$   
  \EndIf
  \Else \ \Return $null$     \label{lineoerror4}
  \EndIf
\ElsIf{$x^i_j$ is an external input, according to $G''$}    \label{lineoerror1}
  \If {$\exists (i,(y^i_j,null),(z^i_j,null)) \in M$}
  \State $e^i_j \leftarrow z^i_j$
  \Else \ \Return $null$\label{lineoerror2}
  \EndIf
\EndIf
\EndFor 
\ForAll {$i \in \{1,...,m\}$}
  \State $s_i=FHE.Eval(hpk,U_i,u,E_{C'_i})$
  \Comment recall that $U=(U_1,U_2,\ldots,U_m)$ 
\EndFor
\If {$[s_1s_2...s_m] \neq v$} \Return $null$    \label{lineoerror6}
  \Else
  \State $e^i \leftarrow e^i_1e^i_2...e^i_m$
  \State store $(i,(u,v),(e^i,T_i(e^i)))$ in $M$ \label{lineofunc2}
  \EndIf
  \If {$T_i(e^i) \neq (\bot,\bot)$ and is not an external output} \Return $\top$ \label{lineofunc3}
  \Else \ \Return the second component of $T_i(e^i)$
  \EndIf \label{lineofunc4}
\EndIf
\end{algorithmic}
\label{oracle}
\end{algorithm}

\begin{algorithm}
\caption{$O_2(T'_1,\ldots,T'_p)$}
\begin{algorithmic}
  \If {$T'_1,\ldots,T'_p$ form a path $P=T'_1 \rightarrow T'_2 \rightarrow \ldots \rightarrow T'_p$}
  \State Generate an external input $X$ to the table graph $G$ such that the evaluation of $G(X)$ includes the evaluation of tables $T_1,\ldots,T_p$. 
  \Comment $X$ is generated as in Algorithm \ref{Path}.
  \State \Return X
  \Else
  \ \Return $null$
  \EndIf
\end{algorithmic}
\label{oracle2}
\end{algorithm}

First, we construct an experiment $Exp$ (see Figure \ref{exp}), which is slightly different to $Exp^{real}$ in Definition \ref{security3}. In $Exp$, the queries of $A_2$ are not answered by calls to $VS.Encode$ and $VS.Path$, but, instead, they are answered by calls to $S_2^{O_1}$ and $S_3^{O_2}$ (recall that $S_2^{O_1}$ and $S_3^{O_2}$ together are the simulator $S_2^O$). We use the shorthands $O'_1$ and $O'_2$ for $S_2^{O_1}$ and $S_3^{O_2}$ in $Exp$ respectively. 

\begin{figure}[h!]
\centering
{\noindent\begin{minipage}[t]{\textwidth}
\begin{center}
\begin{tabular}{|l|}
\hline
    $Exp(1^K)$: \\ \hline
    1. $(G,CP,state_A) \leftarrow A_1(1^K)$ \\
    2. $(G',hpk,U) \leftarrow VS.Encrypt(1^K,G)$  \\
    3. $a \leftarrow A^{O'_1,O'_2}_2(1^K,G',G,CP,hpk,U,state_A)$  \\
    4. Output $a$  \\ \hline
\end{tabular}
\end{center}
\end{minipage}}
\caption{Experiment $Exp$}  \label{exp}
\end{figure}

\begin{lemma}\label{first}
Experiments $Exp^{real}$ and $Exp$ are computationally indistinguishable.
\end{lemma} 
\begin{proof}
First, note that $VS.Path$ and $O'_2$ have the same functionality (see Algorithms \ref{vspath} and \ref{oracle2}). We prove that $VS.Encode$ and $O'_1$ have the same input-output functionality. By the construction of $VS.Encode$ (see lines 1 - \ref{lineqkindencode} in Algorithm \ref{encode}) and $O'_1$ (see lines 1 - \ref{lineqkindo} in Algorithm \ref{oracle}), when a query is of the form $(i,u,v,q_1)$, both algorithms will output the FHE encryption of $u$, so they have the same input-output functionality. The case of a query $(i,u,v,q_2)$ is more complex; there are two cases: 

\noindent{\bf Case 1: Algorithm \ref{encode} outputs $null$.} This happens in the following cases:
\begin{enumerate}
\item An intermediate input $x^i_j \in u$ which should be the output of $T'_k$, is not $T'_k$'s output (see line \ref{lineerror1}).
\item An intermediate input $x^i_j \in u$ is the output of a $T'_k$, but $T'_k$'s output is $(\bot,\bot)$ (see line \ref{lineerror2}).
\item The FHE encryption of an external input $x^i_j \in u$ cannot be found in $VS.Encode$'s memory $M$, where it should have been if it had already been processed (as it should) by $VS.Encode$ (see line \ref{lineerror3}).
\item $T'_i(u) \neq v$ (see line \ref{lineerror4}).
\end{enumerate}
These four cases will also cause Algorithm \ref{oracle} to output $null$ in lines \ref{lineoerror3}-\ref{lineoerror4} (for the first and second), lines \ref{lineoerror1} - \ref{lineoerror2} (for the third), and line \ref{lineoerror6} (for the fourth).

\noindent{\bf Case 2: Algorithm \ref{encode} doesn't output $null$.} Suppose $X$ is the external input to $G'$ and $x^i$ is the input to $T'_i$. There are three cases for $VS.Encode(i,x^i,T'_i(x^i),q_2)$ (see lines \ref{lineencode1} - \ref{lineencode2} in Algorithm \ref{encode}):
\begin{enumerate}
\item If $T'_i$'s output is an intermediate output and the FHE encryption of $(\bot,\bot)$, then $VS.Encode$ decrypts $T'_i(x^i)$, and outputs $\bot$.
\item If $T'_i$'s output is an intermediate output and not the FHE encryption of $(\bot,\bot)$, then $VS.Encode$ decrypts $T'_i(x^i)$, and outputs $\top$.
\item If $T'_i$'s output is an external output, then $VS.Encode$ decrypts $T'_i(x^i)$, and outputs the second component of $FHE.Dec(T'_i(x^i))$ (i.e., the actual output).
\end{enumerate}
On the other hand, when $O'_1(i,x^i,T'_i(x^i),q_2)$ calculates $T_i(u^i)$ at $u^i$ (the unencrypted $x^i$) in lines \ref{lineofunc1} - \ref{lineofunc2} of Algorithm \ref{oracle}, there are three cases (see lines \ref{lineofunc3} - \ref{lineofunc4} in Algorithm \ref{oracle}):
\begin{enumerate}
\item If $T'_i$'s output is an intermediate output and $T_i(u^i)$ is $(\bot,\bot)$, the output is $\bot$.
\item If $T'_i$'s output is an intermediate output and $T_i(u^i)$ is not $(\bot,\bot)$, the output is $\top$.
\item If $T'_i$'s output is an external output, the output is the second component of $T_i(u^i)$.
\end{enumerate} 
The first two cases are the same for $VS.Encode(i,x^i,T'_i(x^i),q_2)$ and $O'_1(i,x^i,T'_i(x^i),q_2)$. In the third case, Algorithm \ref{encode} outputs the second half of $FHE.Dec(T'_i(x^i))$, while Algorithm \ref{oracle} outputs the second component of $T_i(u^i)$, which are the same because of  \eqref{eqinduction}. Therefore, $VS.Encode$ and $O'_1$ have the same input-output functionality.
\end{proof}
 
In order to prove Theorem \ref{thm1}, we first prove the security of a single table:
\begin{lemma}    \label{abe2}
For every p.p.t. adversary $A=(A_1,A_2)$, and any table $T \in G$, consider the following two experiments:
\begin{figure}[h!]
\centering
  \subfloat[Experiment $SingleT^{real}$]{%
    {\noindent\begin{minipage}[t]{\textwidth}
    \begin{center}
    \begin{tabular}{|p{0.01\textwidth}p{0.9\textwidth}|}
    \hline
    & $SingleT^{real}(1^K)$ \\ \hline
    1. &$(G,CP,state_A) \leftarrow A_1(1^K)$. \\
    2. &$(hpk,hsk) \leftarrow FHE.KeyGen(1^K)$  \\ 
    3. &Let $C$ be a circuit that computes $T$'s function. \\
    4. &Generate universal circuit $U=(U_1,\ldots,U_m)$ and string $S_{C}$, such that $U(S_{C},x)=C(x)$.  \\ 
    5. &$E_{C} \leftarrow FHE.Enc(hpk,S_{C})$ \\
    6. &$ a \leftarrow A_2(1^K,(True,E_{C}),G,CP,hpk,U,state_A)$   \\   
    7. &Output $a$  \\ \hline
    \end{tabular}
    \end{center}
    \end{minipage}}} \\
  \subfloat[Experiment $SingleT^{ideal}$]{%
    {\noindent\begin{minipage}[t]{\textwidth}
    \begin{center}
    \begin{tabular}{|p{0.01\textwidth}p{0.9\textwidth}|}
    \hline
    &$SingleT^{ideal}(1^K)$ \\ \hline
    1. &$(G,CP,state_A) \leftarrow A_1(1^K)$.  \\
    2. &$(hpk,hsk) \leftarrow FHE.KeyGen(1^K)$ \\
    3. &Let $\tilde{G}$ be produced by Step 1 of $S_1(1^K)$  \\
    4. &Construct circuit $\tilde{C}$ that computes the function of $\tilde{T} \in \tilde{TS}$.   \\
    5. &Generate universal circuit $U=(U_1,...,U_m)$ and string $S_{\tilde{C}}$, such that $U(S_{\tilde{C}},x)=\tilde{C}(x)$  \\
    6. &$E_{\tilde{C}} \leftarrow FHE.Enc(hpk,S_{\tilde{C}})$   \\
    7. &$a  \leftarrow A_2(1^K,(True,E_{\tilde{C}}),G,CP,hpk,U,state_A)$ \\
    8. &Output $a$  \\ \hline
    \end{tabular}
    \end{center}
    \end{minipage}}}   
    \caption{Experiments $SingleT^{real}$ and $SingleT^{ideal}$}  \label{experi}
\end{figure}

The outputs of the two experiments are computationally indistinguishable.
\end{lemma}
\begin{proof} 
If $A_2$'s inputs in $SingleT^{real}$ and $SingleT^{ideal}$ are computationally indistinguishable, then $A_2$'s outputs are  also computationally indistinguishable. The former is true, because $E_C$ and $E_{\tilde{C}}$ are two FHE ciphertexts, and, therefore, they are computationally indistinguishable under the IND-CPA security of FHE.        
\end{proof}

We generate a sequence of $n+1$ different table graphs, each differing with its predecessor and successor only at one table:
\[ TS^i=(T'_1,\ldots,T'_{i-1}, T'_i, \tilde{T}_{i+1},\ldots,\tilde{T}_n) \text{ for } i=0,1,\ldots,n. \]      
All these new table graphs have the same structure graph $G_{struc}$. To each $TS^i,\ i\in\{0,\ldots,n\}$ corresponds the experiment $Exp^{i}$ in Figure \ref{expi}.
\begin{figure}[h!]
\centering
{\noindent\begin{minipage}[t]{\textwidth}
\begin{center}
\begin{tabular}{|p{0.01\textwidth}p{0.8\textwidth}|}
\hline
  &$Exp^i(1^K)$ \\ \hline
  1. &$(G,CP,state_A) \leftarrow A_1(1^K)$ \\
  2. &$(hpk,hsk) \leftarrow FHE.KeyGen(1^K)$  \\
  3. &Generate universal circuit $U=(U_1,\ldots,U_m)$ \\   
  4. &For $j=1,\ldots,i$  \\
  &\tabitem Construct $C_j$ with an $m$-bit output and computes the function of $T_j \in G$. \\
  &\tabitem Generate string $S_{C_j}$ such that $U(S_{C_j},x)=C_j(x)$.  \\
  &\tabitem $E_{C_j} \leftarrow FHE.Enc(hpk,S_{C_j})$  \\
  5. &Let $\tilde{G}$ be produced by Step 1 of $S_1(1^K)$  \\
  6. &For $j=i+1,...,n$   \\
  &\tabitem Construct circuit $\tilde{C}_j$ computing the function of $\tilde{T_j}\in \tilde{G}$. \\
  &\tabitem Generate string $S_{\tilde{C}_j}$ such that $U(S_{\tilde{C}_j},x)=\tilde{C}_j(x)$.  \\
  &\tabitem $E_{\tilde{C}_j} \leftarrow FHE.Enc(hpk,S_{\tilde{C}_j})$   \\ 
  7. &$G^i \leftarrow [\{(True,E_{C_1}),\ldots,(True,E_{C_i}),(True,E_{\tilde{C}_{i+1}}),\ldots,(True,E_{\tilde{C}_{n}})\},G_{struc}]$  \\
  8. &$ a \leftarrow A_2^{O'_1,O'_2}(1^K,G^i,G,CP,hpk,U,state_A)$ \\
  9. &Output $a$ \\ \hline
\end{tabular}
\end{center}
\end{minipage}}
\caption{Experiment $Exp^i$}  \label{expi}
\end{figure}

Note that $Exp^0$ is the same experiment as $Exp^{ideal}$, since Step 5 is the first step of $S_1$ and Steps 2,3,6,7 are doing
exactly what the second step of $S_1$ does (i.e., $VS.Encrypt$). Also note that $Exp^{n}$ is the same as $Exp$ in Figure \ref{exp}, since the $\tilde{G}$ part of $Exp^{n}$ is ignored, and $G^n$ is the results of $VS.Encrypt(1^K,G)$, i.e., $G^n=G'$.   
   
Now we are ready to prove that $Exp^{real}$ is computationally indistinguishable from $Exp^{ideal}$ by contradiction. Assume that $Exp^{real}$ and $Exp^{ideal}$ are computationally distinguishable, and, therefore, $Exp^0$ and $Exp^n$ are computationally distinguishable, i.e., there is a pair of p.p.t. adversaries $A=(A_1,A_2)$ and a p.p.t. algorithm $D$ such that
\begin{equation}   \label{eq:contra}
\left|Pr[D(Exp^{0}(1^K))=1]-Pr[D(Exp^{n}(1^K))=1]\right|>negl(K).
\end{equation}
Since 
\begin{multline*}
\left|Pr[D(Exp^{0}(1^K))=1]-Pr[D(Exp^{n}(1^K))=1]\right| \leq \\ \sum_{i=0}^{n-1} \left|Pr[D(Exp^{i}(1^K))=1]-Pr[D(Exp^{i+1},(1^K))=1]\right|
\end{multline*} 
inequality \eqref{eq:contra} implies that there exists $0\leq i\leq n-1$ such that
\[  \left|Pr[D(Exp^{i}(1^K))=1]-Pr[D(Exp^{i+1}(1^K))=1]\right|>negl(K)/n=negl(K),  \]
We use $A=(A_1,A_2)$ to construct a pair of p.p.t. algorithms $A'=(A'_1,A'_2)$ which together with p.p.t. algorithm $D$ contradict Lemma \ref{abe2}, by distinguishing $SingleT^{real}$ to $SingleT^{ideal}$ . Specifically, they determine whether Step 6 of $SingleT^{real}$ or Step 7 of $SingleT^{ideal}$ has been executed. $A'_2$ can distinguish the two experiments, if it can distinguish between its two potential inputs $(1^K, (True,E_{C_{i+1}}),U,hpk,state_{A'})$ and $(1^K, (True,E_{\tilde{C}_{i+1}}),U,hpk,state_{A'})$. The idea is to extend the table $(True,E_{C_{i+1}})$ or $(True,E_{\tilde{C}_{i+1}})$ (whichever the case) into a full table graph $G^i$ or $G^{i+1}$ (whichever the case), appropriate for experiments $Exp^i$ or $Exp^{i+1}$ (whichever the case), and invoke $A_2$ which can distinguish between the two. $A'$ is described in Figure \ref{A'}, where the table graph $H$ is either $G^i$ or $G^{i+1}$. Hence, by the construction of $A'$, we know that $D$ can be used to distinguish between experiments $SingleT^{real}$ and $SingleT^{ideal}$ for $T_{i+1}$.

\begin{figure}[h!]
\centering
  \subfloat[Algorithm $A'_1$]{%
    {\noindent\begin{minipage}[t]{\textwidth}
    \begin{center}
    \begin{tabular}{|p{0.01\textwidth}p{0.9\textwidth}|}
    \hline
    &$A'_1(1^K)$ \\ \hline
    1. &$(G,CP,state_A) \leftarrow A_1(1^K)$  \\
    2. &Construct a circuit $C_{i+1}$, and a string $S_{C_{i+1}}$ such that $U(S_{C_{i+1}},x)=C_{i+1}(x)$ \\
    3. &$E_{C_{i+1}} \leftarrow FHE.Enc(hpk,S_{C_{i+1}})$   \\
    4. &$T_{i+1} \leftarrow (True,E_{C_{i+1}})$  \\
    5. &Output $(G,CP,state_{A'})$ \\ \hline
    \end{tabular}
    \end{center}
    \end{minipage}}} \\
  \subfloat[Algorithm $A'_2$]{%
    {\noindent\begin{minipage}[t]{\textwidth}
    \begin{center}
    \begin{tabular}{|p{0.01\textwidth}p{0.9\textwidth}|}
    \hline
    &$A'_2(1^K)$ \\ \hline
    1. &For $j=1,\ldots,i$  \\
    &\tabitem Construct circuit $C_j$ that computes the function of $T_j \in G$. ($G$ comes from $A'_1$) \\
    &\tabitem Generate string $S_{C_j}$ such that $U(S_{C_j},x)=C_j(x)$.   \\ 
    &\tabitem $E_{C_j} \leftarrow FHE.Enc(hpk,S_{C_j})$  \\ 
    &\tabitem Construct table $(True,E_{C_j})$   \\
    2. &Construct $\tilde{G}=[\tilde{TS},G_{struc}]$ just as $S_1$ does.  \\
    3. &For $j=i+2,\ldots,n$  \\
    &\tabitem Construct circuit $\tilde{C}_j$ that computes the function of $\tilde{T}_j \in \tilde{G}$ \\
    &\tabitem Generate string $S_{\tilde{C}_j}$ such that $U(S_{\tilde{C}_j},x)=\tilde{C}_j(x)$.  \\ 
    &\tabitem $E_{\tilde{C}_j} \leftarrow FHE.Enc(hpk,S_{\tilde{C}_j})$   \\  
    &\tabitem Construct table $(True,E_{\tilde{C}_j})$  \\ 
    4. &$H = [\cup_{l=1}^i(True,E_{C_l})\cup [(True,E_{C_{i+1}}) \text{ or } (True,E_{\tilde{C}_{i+1}})] \cup_{l=i+2}^n (True,E_{\tilde{C}_l}), G_{struc}]$   \\
    5. &$ a \leftarrow A_2^{O'_1,O'_2}(1^K, H,G,CP,hpk,U,state_{A'})$   \\
    \hline
    \end{tabular}
    \end{center}
    \end{minipage}}}   
    \caption{Algorithms $A'=(A'_1,A'_2)$}  \label{A'}
\end{figure}
     
\end{proof}

Theorems \ref{thm3} and \ref{thm1} imply
\begin{theorem}   \label{main}
Our verification scheme satisfies Definition \ref{defenc}.
\end{theorem}


\section{Secure and trusted verification for general developers}  \label{sec:maldev}

In general, the developer may not comply with Definition \ref{honest}, i.e., the developer can actually replace $VS.Encode$ with some other malicious algorithm $VS.Encode'$ in its interaction with the verifier. If we do not provide a method to prevent this scenario from happening, then a buggy implementation could pass the verifier's verification when it actually should not. 

Bearing this in mind, we replace our old definition of a verification scheme $VS$ with a new one in Definition \ref{generalvs2} below, by adding an algorithm $VS.Checker$, which the verifier can ask the developer to run in order to determine whether the latter indeed runs $VS.Encode$. $VS.Checker$ itself is also run by the developer, which immediately poses the danger of being replaced by some other algorithm $VS.Checker'$. Therefore, $VS.Checker$ must be designed so that even if it is replaced by some other algorithm, the verifier can still figure out that the developer doesn't run $VS.Checker'$ or $VS.Encode'$ from its replies.

In the new extended definition, $VS.Encrypt$, $VS.Encode$, and $VS.Eval$ remain the same. By running $VS.Eval$ with a publicly known $Certificate$, any third party can check whether the developer is malicious and whether the verification was done correctly. 

\begin{defin}[Extension of Definition \ref{defe}]    \label{generalvs2}
A {\em verification scheme} $VS$ is a tuple of p.p.t. algorithms $(VS.Encrypt, VS.Encode, VS.Checker, VS.Eval)$ such that
\begin{itemize}
\item $VS.Encrypt(1^K, G)$ is a p.p.t. algorithm that takes a security parameter $1^K$ and a table graph $G$ as input, 
and outputs an encrypted table graph $G'$. 
\item $VS.Encode$ is a p.p.t. algorithm that takes an input $x$ and returns an encoding $Enc_{x}$.	
\item $VS.Eval$ is a p.p.t. algorithm that takes a security parameter $K$ and a public $Certificate$ as input, has an honest verifier $V$ satisfying Definition \ref{verfier} hardcoded in it, and outputs $1$ if the verification has been done correctly (and $0$ otherwise).  
\item $VS.Checker$ is a p.p.t. algorithm with a memory $state_C$ that receives a question $Q$ from the verifier and replies with an answer $A$, so that the verifier can detect whether the developer indeed runs $VS.Encode$ and $VS.Checker$.
\end{itemize}
\end{defin}

Figure \ref{protocol3} shows what should be a round of the protocol between the developer and the verifier.
\begin{figure}[h!]
\begin{center}
\includegraphics[width=\textwidth]{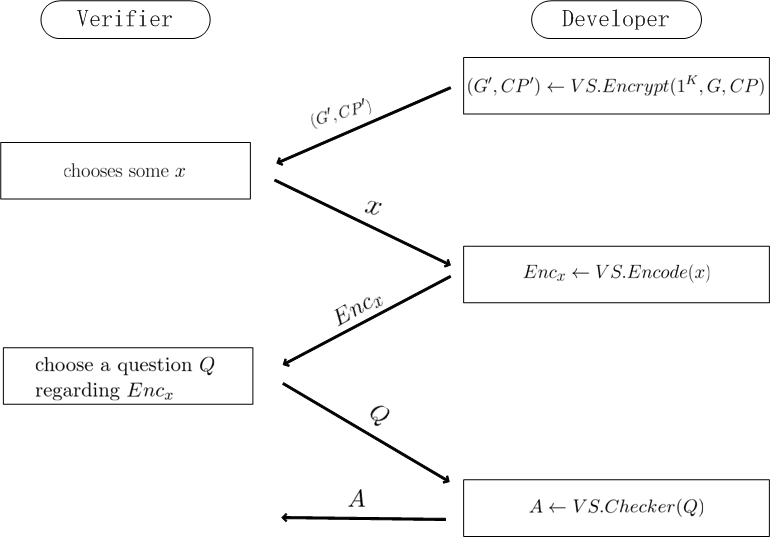}
\end{center}
\caption{The generic $VS$ protocol in Definition \ref{generalvs2}.}
\label{protocol3}
\end{figure}

\begin{defin}[Correctness]   \label{correctness4}
A verification scheme is {\em correct} iff the following holds: 
$VS.Eval(1^K,Certificate)=1$ if and only if all of the following hold:
\begin{eqnarray}
Pr_r[V(1^K,G',G^{spec},VGA_r,CP)(r) = V'(1^K,G',G^{spec},VGA_r,CP))(r)] \geq 1-negl(K), & \label{eqcorrect1}  \\
\forall x_1,x_2:\ Pr[VS.Encode(x_1)=VS.Encode'(x_1),VS.Checker(x_2)=VS.Checker'(x_2)]\geq 1-negl(K), & \label{eqcorrect2} \\
Pr_r[V(1^K,G,G^{spec},VGA_r,CP)(r) = V(1^K,G',G^{spec},VGA_r,CP)(r)] \geq 1-negl(K), & \label{eqcorrect3} 
\end{eqnarray}
where $V$ is the verifier hardwired in $VS.Eval$, and $G'$ is the table graph produced by $VS.Encrypt$.
\end{defin}

\begin{defin}[Security]   \label{security4}
For $A=(A_1,A_2)$ and $S=(S_1,S_2,S_3)$ which are tuples of p.p.t algorithms, consider the two experiments in Figure \ref{nexprl}.
\begin{figure}[h!]
\centering
  \subfloat[Experiment $Exp^{real}$]{%
    {\begin{minipage}[t]{0.6\textwidth}
    \centering
    \begin{tabular}{|l|}
    \hline
    $Exp^{real}(1^K)$  \\ \hline
    1. $(G,CP,state_A) \leftarrow A_1(1^K)$  \\
    2. $G' \leftarrow VS.Encrypt(1^K,G)$  \\
    3. $a \leftarrow A_2^{VS.Encode,VS.Checker }(1^K,G',G,CP,state_A)$  \\
    4. Output $a$  \\ \hline
    \end{tabular}
    \end{minipage}}} \\
  \subfloat[Experiment $Exp^{ideal}$]{%
    {\begin{minipage}[t]{0.6\textwidth}
    \centering
    \begin{tabular}{|l|}
    \hline
    $Exp^{ideal}(1^K)$ \\ \hline
    1. $(G,CP,state_A) \leftarrow A_1(1^K)$  \\
    2. $\tilde{G} \leftarrow S_1(1^K)$  \\
    3. $a \leftarrow A_2^{S^{O_1}_2,S^{O_2}_3}(1^K,\tilde{G},G,CP,state_A)$   \\
    4. Output $a$  \\ \hline
    \end{tabular}
    \end{minipage}}}  
    \caption{The updated experiments}  \label{nexprl}
\end{figure}

A verification scheme $VS$ is \emph{secure} if there exist a tuple of p.p.t. simulators $S=(S_1,S_2,S_3)$ and oracles $O_1,O_2$ such that for all pairs of p.p.t. adversaries $A=(A_1,A_2)$, the following is true for any p.p.t. algorithm $D$:
\[ \left|Pr[D(Exp^{ideal}(1^K),1^K)=1]-Pr[D(Exp^{real}(1^K),1^K)=1]\right| \leq negl(K), \]
i.e., the two experiments are computationally indistinguishable.
\end{defin}

Correspondingly, we update Definition \ref{defenc}:
\begin{defin}[Secure and trusted verification scheme]   \label{ndefenc} 
A verification scheme of Definition \ref{generalvs2} is {\em secure and trusted} iff it satisfies Definitions \ref{correctness4} and \ref{security4}.
\end{defin}

The updated Definition \ref{honest} becomes
\begin{defin}   \label{malicious}
A developer is {\em honest} iff it always runs $VS.Encode$ and $VS.Checker$. Otherwise it is called {\em malicious}.
\end{defin}

\begin{rem}  \label{remnotoracle2}
$VS.Encode$, $VS.Checker$, $S^{O_1}_2$ and $S^{O_2}_3$ in Step 3 of the two experiments in Figure \ref{nexprl} are {\em not} oracles used by $A_2$. In these experiments, $A_2$ plays the role of a (potentially malicious) verifier and interacts with the developer as shown in Figure \ref{protocol3}. More specifically, in $Exp^{real}$, $A_2$ asks the developer to run $VS.Encode$ or $VS.Checker$ on inputs of its choice (instead of querying an oracle), and receives the answer. In $Exp^{ideal}$, $A_2$ again asks the developer, but, unlike $Exp^{real}$, the latter runs $S^{O_1}_2$ instead of $VS.Encode$ and $S^{O_2}_3$ instead of $VS.Checker$, and provides $A_2$ with the answer. Hence, whenever we say that $A_2$ queries $VS.Encode$, $VS.Checker$, $S^{O_1}_2$ or $S^{O_2}_3$ we mean that $A_2$ asks the developer to run $VS.Encode$, $VS.Checker$, $S^{O_1}_2$ or $S^{O_2}_3$ respectively, and provide the answer. Note that $O_1, O_2$ {\em are} oracles for $S_2, S_3$ respectively. 
\end{rem}
We repeat that in Section \ref{sec:honest} we required that the developer satisfied Definition \ref{security3}, but that developer may not comply with Definition \ref{malicious} in this section.

\subsection{Construction outline}   \label{outlinevs}

$VS.Encrypt$, $VS.Encode$ and $VS.Path$ are exactly the same as in Section \ref{sec:honest}. 

\subsubsection{VS.Checker}    \label{sec:checker}

In order to illustrate how $VS.Checker$ is going to be used, we use Figure \ref{VSExp} as an example of the evaluation of a table graph $G$ (on the left) and its encrypted version $G'$ (on the right).  
 \begin{figure}[h!]
 \begin{center}\includegraphics[width=\textwidth]{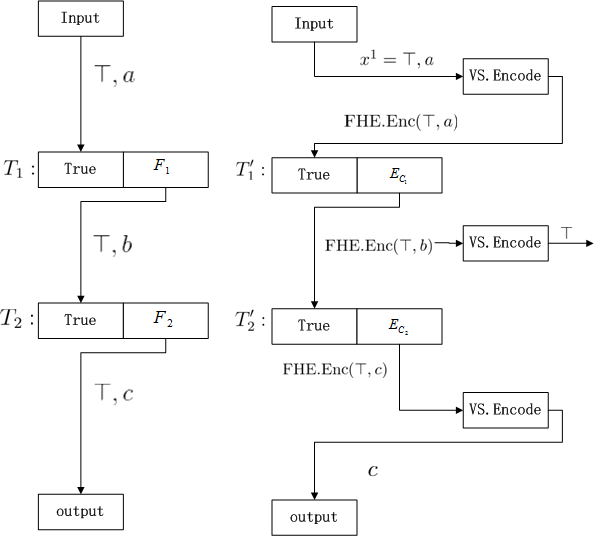}
 \end{center}
 \caption{Example of evaluation of table graph $G$ (left) and its encryption $G'$ (right)}
 \label{VSExp}
 \end{figure}
There are three potential points where a developer can tinker with $VS.Encode$, namely, when the verifier $V$ is querying for an external output (bottom application of $VS.Encode$ in Figure \ref{VSExp}), for an intermediate output (middle application of $VS.Encode$ in Figure \ref{VSExp}), and for an external input (top application of $VS.Encode$ in Figure \ref{VSExp}). 

{\bf Case 1: External output.} We start with this case, since it is somewhat more straight-forward; so, for now, we will assume that the developer has indeed run $VS.Encode$ in the previous tables of the path. Suppose $V$ asks the developer to run $VS.Encode$ to decrypt the external output $FHE.Enc(\top,c)$ of table $T'_2$ in Figure \ref{VSExp}. The correct output of $VS.Encode$ should be the external output $c$, but a malicious developer can replace $VS.Encode$ with a $VS.Encode'$ which outputs $c'\not= c$. $V$ can use the following method to detect this behaviour:
\begin{description}
\item[Step 1] Ahead of running the protocol, $V$ announces the use of a deterministic private key encryption scheme $SE$ (see Definition \ref{symmetrickey}), and chooses a secret key $sk$. ($SE$'s encryption algorithm is $SE.Enc$, represented in a circuit format compatible with $FHE.Eval$.)
\item[Step 2] $V$ first extracts the second component of $FHE.Enc(\top,c)$, which is $FHE.Enc(c)$, and then runs $FHE.Eval$$(hpk$,$SE.Enc$,$FHE.Enc(sk)$,$FHE.Enc(c))$ to get a result $y$.
\item[Step 3] $V$ sends $y$ to the developer and asks it to run $VS.Checker$ in order to fully homomorphically decrypt $y$; the developer returns $VS.Checker$'s output $FHE.Dec(y)$, which we denote as $d$.
\item[Step 4] $V$ uses $SE$'s decryption algorithm $SE.Dec$ to decrypt $d$ and get $SE.Dec(sk,d)$.
\item[Step 5] $V$ compares $SE.Dec(sk,d)$ with the (known) external output $c$. If they are the same, $V$ knows that the developer indeed run $VS.Encode$, otherwise $V$ knows that the developer is malicious.
\end{description}
Obviously, if the developer decides to run $VS.Checker(y)$, i.e., $VS.Checker(FHE.Enc(SE.Enc(sk,c)))$, it gets $d=FHE.Dec(y)=SE.Enc(sk,c)$ which it returns to $V$. Then $V$ obtains $c$ by evaluating $SE.Dec(sk,d)$, and compares it with the answer from the developer who is expected to run $VS.Encode$. If they are not the same, $V$ rejects.
Now suppose that the developer runs some $VS.Checker'$ instead of $VS.Checker$, sending $V$ a value $d'\not= d$, and some $VS.Encode'$, sending $V$ a value $c'\not= c$. Then $V$ uses its secret key $sk$ to decrypt $d'$ and gets $SE.Dec(sk,d')$, which is not $c'$ whp, leading to $V$ rejecting whp.

{\bf Case 2: Intermediate output.} Again, assuming (for now) that the developer has indeed run $VS.Encode$ in the previous tables of the path in Figure \ref{VSExp}, this is the case of $V$ getting the intermediate output $FHE.Enc(\top,b)$. What $V$ is allowed to know from the developer is whether this intermediate output is meaningful or not, namely, whether $FHE.Enc(\top,b)$ contains $\top$ or $\bot$ (cf. Section \ref{sec:construct}). A malicious developer can run $VS.Encode'$ instead of $VS.Encode$ and return $\bot$ instead of $\top$. The method for $VS.Checker$ described in Case 1 can also be used here by changing Step 2, so that $V$ first extracts the first half of $FHE.Enc(\top,c)$ (which is $FHE.Enc(\top)$), and then runs $FHE.Eval(hpk,SE.Enc,FHE.Enc(sk),FHE.Enc(\top))$ to get $y$.

{\bf Case 3: External input.} This is the case of an external input $(\top,a)$ (see top of Figure \ref{VSExp}). The verifier can treat this case exactly like Case 1, to confirm that $FHE.Enc(\top,a)$ is actually a fully homomorphic encryption of $(\top,a)$. 

By doing a consistent traversal of $G_{struc}$, and applying the relevant Case 1-3 each time (i.e., starting with external inputs (Case 3) and working its way through finally the external outputs (Case 1)), $V$ can enforce the developer to run $VS.Encode$ (and $VS.Checker$) in all steps. Unfortunately, allowing $V$ to ask the developer to run $VS.Checker$ actually makes $V$ far more powerful, and there is a risk that $V$ may abuse this power, by sending queries with malicious content to $VS.Checker$. Therefore we need to force $V$ into asking the `right' queries. For example, in Step 3 of Case 1, $VS.Checker$ must check whether $y=FHE.Enc(SE.Enc(sk,c))$ before it decrypts $y$. Our solution to this problem is to use a bit commitment protocol during the interaction between the verifier $V$ and the developer. A description of the interaction between developer and verifier during the execution of such a protocol can be found in Appendix \ref{app:bcp}. Algorithm \ref{vscheck} implements $VS.Checker$. The definition of $QA_E$ is exactly the same as in Section \ref{sec:construct}.

\begin{algorithm}[h!]
\caption{$VS.Checker(i,p,y)$}
\begin{algorithmic}[1]
 \If{$|y| \neq l\cdot m$ or $|p| \neq l\cdot m$}    \label{linevscheckererror1}
 \State  \Return null
 \ElsIf{$\nexists (Q_k,A_k) \in QA_E$ such that $Q_k=(i,x^i,T'_i(x^i),q_2)$ and $p = T'_i(x^i)$} 
 \State \Return null
 \ElsIf{$\nexists (Qe_k,Ae_k) \in QA_E$ such that $(Qe_k,Ae_k) = ((i,u^i_j,null,q_1),w^i_j)$ and $p = w^i_j$}\label{linevscheckererror2}
 \State \Return null
 \Else
 \ForAll {$j \in \{0,...,m-1\}$}   \label{linevschecksedec1}
 \State $b_{j+1} \leftarrow FHE.Dec(hsk,y[j\cdot l+1:(j+1)\cdot l])$   \label{linevschecksedec2}
 \EndFor
 \State The developer starts the bit commitment protocol described in Section \ref{bitcommit1}. The developer wants to commit the verifier to $d=b_1,...,b_m$.  
 \If{(bit commitment protocol failed)}   \label{linevscheckererror3}
  \ \Return null
  \EndIf
  \If{$\exists (Qe_k,Ae_k) \in QA_E$ such that $Qe_k=(i,x^i,T'_i(x^i),q_2)$ and $p = T'_i(x^i)$ and $T'_i$'s output is an intermediate output}    \label{linevscheckererror4}
   \If {$FHE.Eval(hpk,SE.Enc,FHE.Enc(sk),p[1:(m/2\cdot l)]) \neq y$}
   \State \Return $null$
  \EndIf
  \ElsIf {$\exists (Qe_k,Ae_k) \in QA_E$ such that $Qe_k=(i,u^i_j,null,q_1)$ and $p=Ae_k$ and $T'_i$'s input is an external input}
    \If {$FHE.Eval(hpk,SE.Enc,FHE.Enc(sk),p) \neq y$}
     \State \Return $null$
    \EndIf
   \Else  
   \If {$FHE.Eval(hpk,SE.Enc,FHE.Enc(sk),p[(l\cdot m/2+1):(l\cdot m)]) \neq y$}
        \State \Return $null$
       \EndIf
    \EndIf
 \EndIf      \label{linevscheckererror5}
 \State \Return $d$   \label{linevscheck3}
\end{algorithmic}   \label{vscheck}
\end{algorithm}

\subsubsection{VS.Eval}

$VS.Eval$ is an extension of $VS.Eval$ in Section \ref{sec:vseval}. It not only needs to check whether the verifier $V$ evaluates the table graph correctly, but it also needs to check whether the developer replies honestly. In order to check whether the verifier $V$ evaluates the table graph correctly, $VS.Eval$ runs Algorithm \ref{eval}. In addition, by reading the log file which records the interaction between the verifier and the developer, $VS.Eval$ can check whether the verifier actually asks the developer to run $VS.Checker$ for every answer it gets (allegedly) from $VS.Encode$. Actually, $VS.Eval$ can recreate the whole interaction between the verifier and the developer from this log file.

The log file is $QA_E$ (as before) but extended with an extra set $QA_C$ which contains tuples $(Q_i,A_i,S_i)$; each such tuple contains a query $Q_i=(i,p,y)$ by $V$ to $VS.Checker$, the information $S_i$ generated by both the developer and the verifier during the bit commitment protocol, and the answer $A_i$ returned by the developer after running $VS.Checker$. Hence, for each $VS.Encode$ record $(Qe_i,Ae_i) \in QA_E$ generated, a $VS.Checker$ record $(Qc_i,Ac_i,Sc_i)\in QA_C$ will also be generated.
 
Algorithm \ref{VSeval} implements $VS.Eval$, taking logs $QA_E,QA_C$ as its input.

\begin{algorithm}[h!]
  \caption{$VS.Eval(1^K,QA_E,QA_C)$}
  \begin{algorithmic}[1]
  \If {Algorithm \ref{eval} returns 0}    \label{linevsevalimply}
  \Comment Algorithm \ref{eval} is the old $VS.Eval$ (Section \ref{sec:vseval})
  \State \Return 0
  \EndIf
  \ForAll{ $(Qe_i,Ae_i) \in QA_E$}
    \If{$\nexists (Qc_i,Ac_i,Sci) \in QA_C$ corresponding to $(Qe_i,Ae_i)$} 
    \ \Return 0
  \Else
  \ find the corresponding $((i,p,y),d,Sc_i) \in QA_C$   
  \EndIf
  \State Use hardwired honest verifier $V$ and information $Sc_i$ to replay the bit commitment protocol, and check whether $d$ produced by $VS.Checker$ is the value $d'$ the original verifier committed to.
  \If{$d'\not= d$ }
  \Return 0
  \EndIf
  \If{$Qe_i$'s format is $(i,v^i_j,null,q_1)$}
  \If{ $ FHE.Eval(hpk,SE.Enc,FHE.Enc(sk),p) \neq y$} \Return 0
  \ElsIf{$SE.Dec(sk,d) \neq v^i_j$} \Return 0
  \Else \ \Return 1
  \EndIf
  \ElsIf{ $Qe_i$'s format is $(i,x^i,T'_i(x^i),q_2)$ and $T'_i(x^i)$ is an external output}
  \If{ $FHE.Eval(hpk,SE.Enc,FHE.Enc(sk),p[(l\cdot m/2+1):(l\cdot m)]) \neq y$} \Return 0
  \ElsIf {$SE.Dec(sk,d) \neq Ae_i$} \Return 0
  \Else \ \Return 1
  \EndIf
  \ElsIf{ $Qe_i$'s format is $(i,x^i,T'_i(x^i),q_2)$ and $T'_i(x^i)$ is an intermediate output}
  \If{ $FHE.Eval(hpk,SE.Enc,FHE.Enc(sk),p[1:(l\cdot m/2)]) \neq y$} \Return 0
  \ElsIf {$SE.Dec(sk,d) \neq Ae_i$} \Return 0    \label{linevsevaldecrypt}
  \Else \ \Return 1
  \EndIf
  \EndIf
  \EndFor
\end{algorithmic}   \label{VSeval}
\end{algorithm}

\subsection{Correctness and security}
   
In our implementation, the $Cerificate$ used by $VS.Eval$ in Definition \ref{generalvs2} is $(QA_E,QA_C,G',hpk,U)$.
In Definition \ref{security4}, simulators $S_1$, $S_2$ and $S_3$ simulate $VS.Encrypt$, $VS.Encode$ and $VS.Checker$ respectively. Similarly to Section \ref{sectionproof3}, we add one more simulator $S_4$ to simulate $VS.Path$, and its corresponding oracle (Algorithm \ref{oracle2}). Algorithm \ref{oracle3} describes the oracle $O_3$ used by $S_4$ to simulate $VS.Checker$.
   
\begin{algorithm}[h!]
\caption{$O_3(i,p,y)$}
\begin{algorithmic}[1]
   \If{$|y| \neq l\cdot m$ or $|p| \neq l\cdot m$}  \label{lineoracle3error1}
   \State \Return $null$
   \ElsIf{$\nexists (Qe_k,Ae_k) \in QA_E$ such that $Qe_k=(i,x^i,T'_i(x^i),q_2)$ and $p = T'_i(x^i)$}
   \State  \Return null
   \ElsIf{$\nexists (Qe_k,Ae_k) \in QA_E$ such that $(Qe_k,Ae_k) = ((i,u^i_j,null,q_1),w^i_j)$ and $p = w^i_j$}  \label{lineoracle3error2}
   \State \Return $null$
   \Else
   \State Find $(Qe_k,Ae_k)\in QA_E$ that matches the input $(i,p,y)$.   \label{lineorace31}
   \If {$Qe_k$ is $(i,u^i_j,null,q_1)$}
   \State $a \leftarrow u^i_j$.
   \Else  \ $a \leftarrow Ae_k$.
   \EndIf
   \State $b_1b_2...b_{m} \leftarrow SE.Enc(sk,a)$     \label{lineoracle32}  
   \Comment Oracle  $O_3$ knows secret key $sk$ and $SE.Enc$ used by $A$ 
    \State The developer starts the bit commitment protocol described in Section \ref{bitcommit1}. The developer wants to commit the verifier to $d=b_1,...,b_m$.  
 \If{(bit commitment protocol failed)}   \label{lineoracle3error3}
  \ \Return null
  \EndIf
  \If{$\exists (Qe_k,Ae_k) \in QA_E$ such that $Qe_t=(i,x^i,T'_i(x^i),q_2)$ and $p = T'_i(x^i)$ and $T'_i$'s output is an intermediate output}   \label{lineoracle3error4}
      \If {$FHE.Eval(hpk,SE.Enc,FHE.Enc(sk),p[1:(m/2\cdot l)]) \neq y$}
          \State \Return $null$
      \EndIf
      \ElsIf {$\exists (Qe_k,Ae_k) \in QA_E$ such that $Qe_k=(i,u^i_j,null,q_1)$ and $Ae_k = p$ and $T'_i$'s input is an external input}
         \If {$FHE.Eval(hpk,SE.Enc,FHE.Enc(sk),p) \neq y$}
         \State \Return $null$
         \EndIf
       \Else  
          \If {$FHE.Eval(hpk,SE.Enc,FHE.Enc(sk),p[(l\cdot m/2+1):(l\cdot m)]) \neq y$}
          \State \Return $null$
          \EndIf
        \EndIf     \label{lineoracle3error5}
  \EndIf     
  \State \Return $d$   \label{lineoracle35}
\end{algorithmic}    \label{oracle3}
\end{algorithm}

\subsubsection{Correctness}
    
We show the following
\begin{theorem}\label{thmcorrect}
The verification scheme $VS$ introduced in this section satisfies Definition \ref{correctness4}.
\end{theorem}
\begin{proof}
We prove that if $VS.Eval(1^k,QA_E,QA_C)=1,$ then inequalities \eqref{eqcorrect1}-\eqref{eqcorrect3} hold.

\begin{lemma}  \label{lem:cond1}
If $VS.Eval(1^k,QA_E,QA_C)=1$, then \eqref{eqcorrect2} holds. 
\end{lemma}
\begin{proof}
Suppose that $Qc_i=(k,p,y)$ and $Ac_i=d$. For brevity reasons, we will only consider the case of $Qe_i=(k,x^k,T'_k(x^k),q_2)$, and $T'_k(x^k)$ is an external output. The other cases ($Qe_i=(k,x^k_j,null,q_1)$ or   $T'_k(x^k)$ is not an external output), are similar. Also, since the bit commitment protocol succeeds whp, $V$ can be assumed to be honest.  

First we consider the case of $VS.Encode$ and $VS.Encode'$ having the same input-output functionality.
Let $(Qc_i,Ac_i,Sc_i) \in QA_C$ be the tuple that corresponds to the current $VS.Encode$ query $(Qe_i,Ae_i)$.
Suppose that $VS.Checker'$ outputs a value $d$, while $VS.Checker$ would output $d^*$. 
If $VS.Eval(1^k,QA_E,QA_C)=1$, then we know (see line \ref{linevsevaldecrypt} in Algorithm \ref{VSeval})
\begin{equation}\label{eqsedec1}
SE.Dec(sk,d)=Ae_i
\end{equation}  
Therefore, to prove that $d=d^*$, it is enough to prove that
\begin{equation}   \label{eqsedecae}
SE.Dec(sk,d^*)=Ae_i.
\end{equation}
According to $VS.Checker$'s construction (see lines \ref{linevschecksedec1}-\ref{linevschecksedec2} in Algorithm \ref{vscheck}), 
\begin{equation}\label{eqfhedecy}
d^*=FHE.Dec(hsk,y),
\end{equation} 
where $Qc_i=(i,p,y)$. 
Since $VS.Encode$ and $VS.Encode'$ have the same input-output functionality, given $Qe_i=(k,x^k,T'_k(x^k),q_2)$ as input to both $VS.Encode$ and $VS.Encode'$, their outputs are the same, i.e., $Ae^*_i=Ae_i$, where $Ae^*_i$ is the output of $VS.Encode$ and $Ae_i$ is the output of $VS.Encode'$. According to the construction of $VS.Encode$ (see lines \ref{lineencode1}-\ref{lineencode2} in Algorithm \ref{encode}) ,
\begin{equation}   \label{eqfhedecae*}
Ae^*_i=FHE.Dec(hsk,T'_k(x^k)[lm/2+1:lm])
\end{equation}
Therefore, 
\begin{equation}    \label{eqfhedecae}
Ae_i=FHE.Dec(hsk,T'_k(x^k)[lm/2+1:lm]).
\end{equation} 
On the other hand, whp Definition \ref{homomorphism} implies that
\begin{align}
y=& FHE.Eval(hpk,SE.Enc,FHE.Enc(sk),T'_k(x^k)[lm/2+1:lm])  \nonumber\\
  =& FHE.Enc(hpk,SE.Enc(sk,FHE.Dec(hsk,T'_k(x^k)[lm/2+1:lm]))    \label{eqfheevaly}
\end{align}
Hence, by combining \eqref{eqfhedecy} and \eqref{eqfheevaly} we get
\[ d^*=SE.Enc(sk,FHE.Dec(hsk,T'_k(x^k)[lm/2+1:lm]))), \]
which implies that
\begin{align}
SE.Dec(sk,d^*)=& SE.Dec(sk, SE.Enc(sk,FHE.Dec(hsk,T'_k(x^k)[lm/2+1:lm]))))\nonumber\\
    =& FHE.Dec(hsk,T'_k(x^k)[lm/2+1:lm])   \label{eqfhesedecd}
\end{align}
Then, by combining \eqref{eqfhedecae} and \eqref{eqfhesedecd}, \eqref{eqsedecae} holds whp.   

Next we consider the case of $VS.Checker$ and $VS.Checker'$ having the same input-output functionality. 
For $(Qe_i,Ae_i) \in QA_E$, and the corresponding pair $(Qc_i,Ac_i)\in QA_C$, $Qe_i=(k,x^k,T'_k(x^k),q_2)$ is the input of $VS.Encode'$ and $Ae_i$ its output, while $Qc_i=(i,p,y)$ is the input of $VS.Checker'$ and $Ac_i=d$ its output. Since $VS.Eval(1^k,QA_E,QA_C)=1$, we know that \eqref{eqsedec1} holds (see line \ref{linevsevaldecrypt} in Algorithm \ref{VSeval}). Also \eqref{eqfhedecae*} is obviously true, as is \eqref{eqfhedecy} (see lines \ref{linevschecksedec1}-\ref{linevschecksedec2} in Algorithm \ref{vscheck}). The latter, together with the identical functionality of $VS.Checker, VS.Checker'$, implies that
\begin{equation}    \label{eqfhedecd2}
d=FHE.Dec(hsk,y)
\end{equation}
Hence, by combining \eqref{eqsedec1},\eqref{eqfhedecd2},\eqref{eqfheevaly} we have 
\begin{align*}
Ae_i=& SE.Dec(sk,d) \\
       =& SE.Dec(sk,FHE.Dec(hsk,y)) \\
       =& SE.Dec(sk,SE.Enc(sk,FHE.Dec(hsk,T'_k(x^k)[lm/2+1:lm])) \\
       =& FHE.Dec(hsk,T'_k(x^k)[lm/2+1:lm])
\end{align*}
Hence, \eqref{eqfhedecae} holds, and combined with \eqref{eqfhedecae} and \eqref{eqfhedecae*}, we get $Ae_i=Ae^*_i$, i.e., given $Qe_i$ as input to $VS.Encode$ and $VS.Encode'$, their outputs are the same whp. 
     
Finally suppose that there exists $(Qe_i,Ae_i)\in QA_E$ and corresponding $(Qc_i,Ac_i,Sc_i)\in QA_C$ such that $VS.Checker(Qc_i) \neq VS.Checker'(Qc_i)$ and $VS.Encode(Qe_i) \neq VS.Encode'(Qe_i)$. We know (see lines \ref{lineencode1}-\ref{lineencode2} in Algorithm \ref{encode}) that \eqref{eqfhedecae*} holds, and, by combining \eqref{eqfhedecy}, \eqref{eqfhedecae*} and \eqref{eqfheevaly}, we get
\begin{align*}
   SE.Dec(sk,d^*)=& SE.Dec(sk,FHE.Dec(hsk,y))  \\
   =& SE.Dec(sk,SE.Enc(sk,FHE.Dec(hsk,T'_k(x^k)[lm/2+1:lm])))  \\
   =& FHE.Dec(hsk,T'_k(x^k)[lm/2+1:lm]) = Ae^*_i \label{eqsedecae*}
\end{align*}
Since $d^* \neq d$ and $Ae^*_i \neq Ae_i$, $d$ and $Ae_i$ do not satisfy \eqref{eqsedec1} whp. But according to $VS.Eval$'s construction (see line \ref{linevsevaldecrypt} in Algorithm \ref{VSeval}), when $VS.Eval(1^k,QA_E,QA_C)=1$, \eqref{eqsedec1} holds whp, a contradiction. 
\end{proof} 
\begin{lemma}    \label{lem:cond2}
If $VS.Eval(1^k,QA_E,QA_C)=1$, then \eqref{eqcorrect1} and \eqref{eqcorrect3} hold.
\end{lemma}
\begin{proof}
Since according to $VS.Eval$'s construction (see line \ref{linevsevalimply} in Algorithm \ref{VSeval}), Algorithm \ref{VSeval} outputting 1 implies Algorithm \ref{eval} outputting 1, \eqref{eqcorrectness1} and \eqref{eqcorrectness2} must hold, and they continue to hold while Algorithm \ref{VSeval} invokes $VS.Checker$. Together with \eqref{eqcorrect2} (which we have already proven), \eqref{eqcorrect1} and \eqref{eqcorrect3} easily follow.
\end{proof}    
\end{proof}

\subsubsection{Security}  \label{sec:gensecure}
   
Following the methodology of \cite{goldwasser2013reusable}, we will show the following
\begin{theorem}   \label{thmsecurity}
The verification scheme VS introduced in this section satisfies Definition \ref{security4}.
\end{theorem}
\begin{proof}
We construct a tuple of simulators $(S_1,S_2,S_3,S_4)$ so that $S_1$,$S^{O_1}_2$ and $S^{O_2}_3$ are the same as in the proof of Theorem \ref{thm1}. $S_4$ receives queries from $A_2$, queries oracle $O_3$ (Algorithm \ref{oracle3}), and returns the output of $O_3$ to $A_2$. 

In our proof, we will need to define a new experiment $Exp^{extra}(1^K)$ (cf. Figure \ref{expextra}).
\begin{figure}[h!]
\centering
{\noindent\begin{minipage}[t]{\textwidth}
\begin{center}
\begin{tabular}{|p{0.01\textwidth}p{0.7\textwidth}|}
\hline
   &$Exp^{extra}(1^K)$ \\ \hline
   1. &$(G,CP,state_A) \leftarrow A_1(1^K)$  \\
   2. &$(G',hpk,U) \leftarrow VS.Encrypt(1^K,G)$  \\
   3. &$a \leftarrow A_2^{VS.Encode,VS.Path,\ S^{O_3}_4 }(VS.Eval,G',G,VGA,CP,hpk,state_A)$  \\
   4. &Output $a$ \\ \hline
\end{tabular}
\end{center}
\end{minipage}}
\caption{Experiment $Exp^{extra}$}  \label{expextra}
\end{figure}
$Exp^{extra}(1^K)$ and $Exp^{real}(1^K)$ (cf. Figure \ref{nexprl}) differ only in Step 3, where $A_2$ queries $VS.Checker$ in $Exp^{real}(1^K)$ and $S^{O_3}_4$ in $Exp^{extra}(1^K)$. If we can show that $VS.Checker$ and $S^{O_3}_4$ have the same input-output functionality, then $Exp^{real}(1^K)$ and $Exp^{extra}(1^K)$ are computationally indistinguishable. 

To prove this we distinguish two cases for an input $(i,p,y)$ to $VS.Checker$ and $O_3$. 
 
First, when $VS.Checker(i,p,y)$ outputs $null$, it is easy to see that $O_3(i,p,y)$ will also output $null$; this happens when the size of $p$ or $y$ is not correct (line \ref{linevscheckererror1} in Algorithm \ref{vscheck} and \ref{lineoracle3error2} in Algorithm \ref{oracle3}), when $p$ is not generated by the evaluation of a table or $VS.Encode$ (line \ref{linevscheckererror2} in Algorithm \ref{vscheck} and line \ref{lineoracle3error3} in Algorithm \ref{oracle3}), and when the bit commitment protocol fails (lines \ref{linevscheckererror3}-\ref{linevscheckererror5} in Algorithm \ref{vscheck} and lines \ref{lineoracle3error3}-\ref{lineoracle3error5} in Algorithm \ref{oracle3}).
    
Second, when $VS.Checker(i,p,y)$ does not output $null$, it will output a value $d$;  in this case $O_3(i,p,y)$ first generates the same $d$ (see lines \ref{lineorace31}-\ref{lineoracle32} of Algorithm \ref{oracle3}), and after passing the bit commitment protocol, it also outputs $d$. 

Therefore $VS.Checker$ and $S^{O_3}_4$ have the same input-output functionality. Hence $Exp^{real}(1^K)$ and $Exp^{extra}(1^K)$ are computationally indistinguishable:
\begin{equation}\label{eq:3}
|Pr[D(Exp^{real}(1^K),1^K)=1]-Pr[D(Exp^{extra}(1^K),1^K)=1]| \leq negl(K)
\end{equation} 
   
$Exp^{rtest}(1^K)$ and $Exp^{itest}(1^K)$ in Figure \ref{nexprli} are the experiments $Exp^{real}(1^K)$ and $Exp^{ideal}(1^K)$ in Definition \ref{security3}, with $VS.Path$ added to $Exp^{real}(1^K)$ and $S^{O_2}_3$ added to $Exp^{ideal}(1^K))$ (see Figure \ref{exprl}). 
   
\begin{figure}[h!]
\centering
  \subfloat[Experiment $Exp^{rtest}(1^K)$]{%
    {\noindent\begin{minipage}[t]{0.6\textwidth}
    \begin{center}
    \begin{tabular}{|l|}
    \hline
   $Exp^{rtest}(1^K)$  \\ \hline
   1. $(G,CP,state_A) \leftarrow A_1(1^K)$  \\
   2. $(G',hpk,U) \leftarrow VS.Encrypt(1^K,G)$   \\
   3. $a \leftarrow A_2^{VS.Encode,VS.Path}(G',G,CP,hpk,U,state_A)$  \\
   4. Output $a$   \\ \hline
    \end{tabular}
    \end{center}
    \end{minipage}}} \\
  \subfloat[Experiment $Exp^{itest}(1^K)$]{%
    {\noindent\begin{minipage}[t]{0.48\textwidth}
    \centering
    \begin{tabular}{|l|}
    \hline
   $Exp^{itest}(1^K)$ \\ \hline
   1. $(G,CP,state_A) \leftarrow A_1(1^K)$  \\
   2. $(G'',hpk,U) \leftarrow S_1(1^K, s,d,G_{struc})$   \\
   3. $a \leftarrow A_2^{S^{O_1}_2,S^{O_2}_3}(G'',G,CP,hpk,U,state_A)$   \\
   4. Output $a$  \\ \hline
    \end{tabular}
    \end{minipage}}}   
    \caption{Experiments $Exp^{rtest}$ and $Exp^{itest}$}  \label{nexprli}
\end{figure}
In the same way as in Theorem \ref{thm1}, we can show that $Exp^{rtest}(1^K)$ and $Exp^{itest}(1^K)$ are computationally indistinguishable, i.e., for all pairs of p.p.t. adversaries $A=(A_1,A_2)$ and any p.p.t. algorithm $D$,
\begin{equation}   \label{eq:2}
|Pr[D(Exp^{rtest}(1^K),1^K)=1]-Pr[D(Exp^{itest}(1^K),1^K)=1]| \leq negl(K)
\end{equation}
   
Assume that $Exp^{real}(1^K)$ and $Exp^{ideal}(1^K)$ are computationally distinguishable. Then \eqref{eq:3} implies that $Exp^{extra}(1^K)$ and $Exp^{ideal}(1^K)$ are computationally distinguishable, i.e., there are a p.p.t. algorithm $\tilde{D}$ and a p.p.t. adversary $\tilde{A}=(\tilde{A}_1,\tilde{A}_2)$ such that
\begin{equation}   \label{eq:1}
|Pr[\tilde{D}(Exp^{ideal}(1^K),1^K)=1]-Pr[\tilde{D}(Exp^{extra}(1^K),1^K)=1]|>negl(K)
\end{equation}
   
A p.p.t. adversary $A=(A_1,A_2)$ that wants to distinguish between $Exp^{rtest}(1^K)$ and $Exp^{itest}(1^K)$ can use $\tilde{A}$ and $\tilde{D}$ as follows:
\begin{itemize}
\item $A_1$ runs $\tilde{A}_1$ and outputs $(G,CP,state_{\tilde{A}}) \leftarrow \tilde{A}_1(1^K)$.
\item $A_2$ gets one of the two:
 \begin{itemize}
 \item $(G',hpk,U) \leftarrow VS.Encrypt(1^K,G)$ and has oracle access to $VS.Encode$, or
 \item $(G'',hpk,U) \leftarrow S_1(1^K,s,d ,G_{struc})$ and has oracle access to $S^{O_1}_2,S^{O_2}_3$
 \end{itemize}   
depending on which one of the two experiments ($Exp^{rtest}(1^K)$ and $Exp^{itest}(1^K)$) is executed. Then $A_2$ constructs $S^{O_3}_4$ and $VS.Eval$ (which are also used in $Exp^{ideal}(1^K)$ and $Exp^{extra}(1^K)$), and runs $\tilde{A}_2$ providing it with access to $S^{O_3}_4$. Its output $a$ will be either $\tilde{A}_2^{VS.Encode,VS.Path,S^{O_3}_4 }(VS.Eval,G',G,CP,hpk,U,state_{\tilde{A}})$ or $\tilde{A}_2^{S^{O_1}_2,S^{O_2}_3,S^{O_3}_4 }(VS.Eval,G'',G,CP,hpk,U,state_{\tilde{A}})$, depending again on which of $Exp^{rtest}(1^K)$ and $Exp^{itest}(1^K)$ is being executed.
\end{itemize}
If $Exp^{rtest}(1^K)$ is being executed, then it can easily be seen that 
\[Pr[\tilde{D}(Exp^{rtes}(1^K),1^K)=1]=Pr[\tilde{D}(Exp^{extra}(1^K),1^K)=1. \] 
If $Exp^{itest}(1^K)$ is being executed, then it can easily be seen that 
\[Pr[\tilde{D}(Exp^{itest}(1^K),1^K)=1]=Pr[\tilde{D}(Exp^{ideal}(1^K),1^K)=1.\] 
But then, \eqref{eq:1} implies that 
\[|Pr[\tilde{D}(Exp^{rtest}(1^K),1^K)=1]-Pr[\tilde{D}(Exp^{itest}(1^K),1^K)=1]|>negl(K).\] 
which contradicts \eqref{eq:2}. Therefore, $Exp^{real}(1^K)$ and $Exp^{ideal}(1^K)$ are computationally indistinguishable.  
\end{proof}


\section{Open problems}

We have presented protocols that implement secure and trusted verification of a design, taking advantage of any extra
information about the structure of the design component interconnections that may be available. Although we show the
feasibility of such verification schemes, ours is but a first step, that leaves many questions open for future research.
\begin{itemize}
\item {\bf Improving efficiency} Our implementation uses FHE, which, up to the present, has been rather far from being
implemented in a computationally efficient way. On the other hand, garbled circuits are usually considered to be more 
efficient than FHE schemes; for example, \cite{huang2011faster} shows that garbled circuits were much more efficient
than a homomorphic encryption scheme in certain Hamming distance computations. Therefore, pursuing protocols based
on Yao's garbled circuits is a worthy goal, even if a more efficient garbled circuits construction is less secure. 
\item {\bf Verifiable computing} Although verifiable computing is not yet applicable to our case (as mentioned in the
Introduction), coming up with a method to hide the computation would provide a more efficient solution to the problem
of secure and trusted verification, since the amount of re-computation of results needed would be significantly reduced.  
\item {\bf Hiding the graph structure} Our work has been based on the assumption that the table graph $G_{struc}$
of a design is known. But even this may be a piece of information that the designer is unwilling to provide, since it
could still leak some information about the design. For example, suppose that the design uses an off-the-shelf subdesign
whose component structure is publicly known; then, by looking for this subgraph inside $G_{struc}$, someone can
figure out whether this subdesign has been used or not. In this case, methods of hiding the graph structure, by, e.g.,
node anonymization such as in \cite{zhou2008preserving}, \cite{cheng2010k}, may be possible to be combined with 
our or other methods, to provide more security.    
\item {\bf Public information vs. testing} The extra information we require in order to allow some white-box test case
generation by the verifier, namely the table graph structure, is tailored on specific testing algorithms 
(such as MC/DC \cite{hayhurst2001practical}),
which produce computation paths in that graph. But since there are other possibilities for test case generation, the
obvious problem is to identify the partial information needed for applying these test generation algorithms, and the
development of protocols for secure and trusted verification in these cases.  
\end{itemize}

\bibliographystyle{plain}
\bibliography{references}


\begin{appendices}

\section{Conversion of functions to circuits}   \label{sec:circuits}
 
We have mentioned above that in our construction the rhs functions in the tables of a table graph will be converted into boolean circuits in order for our construction to work properly. But we did not say how the conversion is done or whether this conversion is practical. Essentially what our construction requires is an automatic way to transform a higher level description of a function into a boolean circuit. Transformation of a higher level description of a function into a lower level description like a circuit was done in \cite{malkhi2004fairplay}, \cite{huang2011faster}. ~\cite{malkhi2004fairplay} introduces a system called FAIRPLAY which includes a high level language SFDL that is similar to C and Pascal. FAIRPLAY allows the developer to describe the function in SFDL, and automatically compiles it into a description of a boolean circuit. A similar system is described in \cite{huang2011faster}, and is claimed to be faster and more practical than FAIRPLAY. This system allows the developer to design a function directly in Java, and automatically compiles it into a description of a boolean circuit.

\section{Bit commitment protocols}   \label{app:bcp}

We are going to use the construction of a CMB protocol by Naor \cite{naor1991bit}.

Let $C \subset \{0,1\}^q$ be a code of $|C|=2^m$ words, such that the Hamming distance between any $c_1, c_2 \in C$ is at least $\epsilon q$, for some $\epsilon>0$. Let $E$ be an efficiently computable encoding function $E:\{0,1\}^m \rightarrow \{0,1\}^q$ that maps $\{0,1\}^m$ to $C$. It is also required that $q\cdot log(2/(2-\epsilon)) \geq 3K$ and $q/m=c$, where $c$ is a fixed constant. \\
$G$ denotes a pseudo-random generator $G$:$\{0,1\}^K \rightarrow \{0,1\}^{l(K)}$, $l(K)>K$ such that for all p.p.t. adversary A ,
\[|Pr[A(y)=1]-Pr[A(G(s))=1] < 1/p(K),\]
where the probabilities are taken over $y \in \{0,1\}^{l(K)}$ and seed $s \in \{0,1\}^K$ chosen uniformly at random.
$G_k(s)$ denotes the first $k$ bits of the pseudo-random sequence on seed $s \in \{0,1\}^K$ and $B_i(s)$ denotes the $i$th bit of the pseudo-random sequence on seed $s$.
For a vector $\vec R=(r_1,r_2,...,r_{2q})$ with $r_i \in \{0,1\}$ and $q$ indices $i$ such that $r_i=1$, $G_{\vec R}(s)$ denotes the vector $\vec A =(a_1,a_2,...,a_q)$ where $a_i=B_{j(i)}(s)$ and $j(i)$ is the index of the $i$th 1 in $\vec R$. If $e_1,e_2\in \{0,1\}^q$, then $e_1\oplus e_2$ denotes the bitwise Xor of $e_1$ and $e_2$.\\
Suppose Alice commits to $b_1,b_2,...,b_m$.
\begin{itemize}
\item \textbf{Commit Stage:}
\begin{enumerate}
\item  Bob selects a random vector $\vec R=(r_1,r_2,...,r_{2q})$ where $r_i \in \{0,1\}$ for $1 \leq i \leq 2q$ and exactly $q$ of the $r_i$'s are 1, and sends it to Alice. 
\item Alice computes $c=E(b_1,b_2,...,b_m)$, selects a seed $S \in \{0,1\}^n$ and sends to Bob $Enc_D$ which is the following: Alice sends Bob $e=c \oplus G_{\vec R}(s)$ (the bitwise Xor of $G_{\vec R}(s)$ and $c$), and for each $1 \leq i \leq 2q$ such that $r_i=0$ she sends $B_i(s)$. 
\end{enumerate}
\item \textbf{Reveal Stage:}
Alice sends s and $b_1,b_2,...,b_m$. Bob verifies that for all $1 \leq i \leq 2q$ such that $r_i=0$, Alice has sent the correct $B_i(s)$, computes $c=E(b_1,b_2,...,b_m)$ and $G_{\vec R}(s)$, and verifies that $e=c \oplus G_{\vec R}(s)$.
\end{itemize}

We denote $\vec R$ by $R$, the bits $b_1,b_2,...,b_m$ by $D$, the seed $s$ by $Cert$; a schematic diagram of the protocol can be found in Figure \ref{bitcommittable}.
\begin{figure}[h!]
\begin{center}
\includegraphics[width=0.48\textwidth]{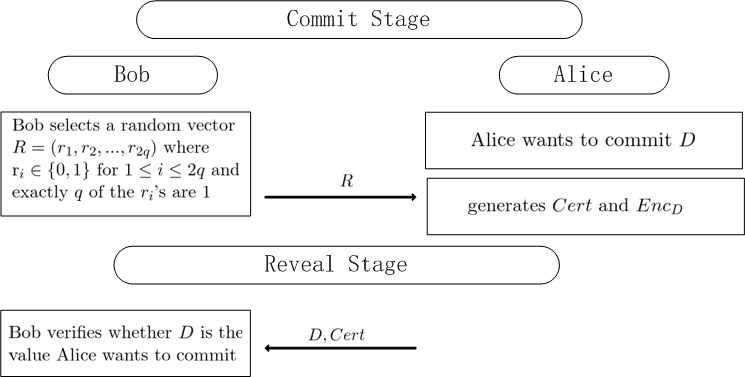}
\end{center}
\caption{The protocol of Naor's construction~\cite{naor1991bit}.}
\label{bitcommittable}
\end{figure}

\subsection{Use of BCP in Section \ref{outlinevs}}

The idea of a bit commitment protocol is similar to hiding information in an envelope \cite{kilian1989uses}. The developer runs $VS.Checker$ which does the following: First, given an input $Q_i=(i,p,y)$, $VS.Checker$ fully homomorphically decrypts $y$ and gets a value $d$. The developer puts $d$ in a sealed envelope and gives it to $V$. $V$ cannot open the envelope at this stage. After that, $VS.Checker$ needs $V$ to provide a proof that $y$ is indeed $FHE.Enc(SE.Enc(sk,c))$. The proof contains the secret key $sk$, but revealing $sk$ to $VS.Checker$ does not matter, because $VS.Checker$'s output is already given to $V$ and cannot change at this point. If $VS.Checker$ confirms that $y$ is indeed $FHE.Enc(SE.Enc(sk,c))$, then it generates a key to the sealed envelope, and the developer gives this key to $V$. Then $V$ can know $d$, using the key to open the envelope. If $VS.Checker$ figures out that $y$ is not $FHE.Enc(SE.Enc(sk,c))$, it simply refuses to generate the key to the envelope. Such an envelope solves the problem of mutual distrust and potential malicious activities by both entities. A bit commitment protocol can play the role of such an envelope.

\section{The protocol of Section \ref{sec:honest}}

Figure \ref{protocol2} describes what the developer and the verifier should do in Section \ref{sec:honest}.
\begin{figure}[h!]
\begin{center}
\includegraphics[width=\textwidth]{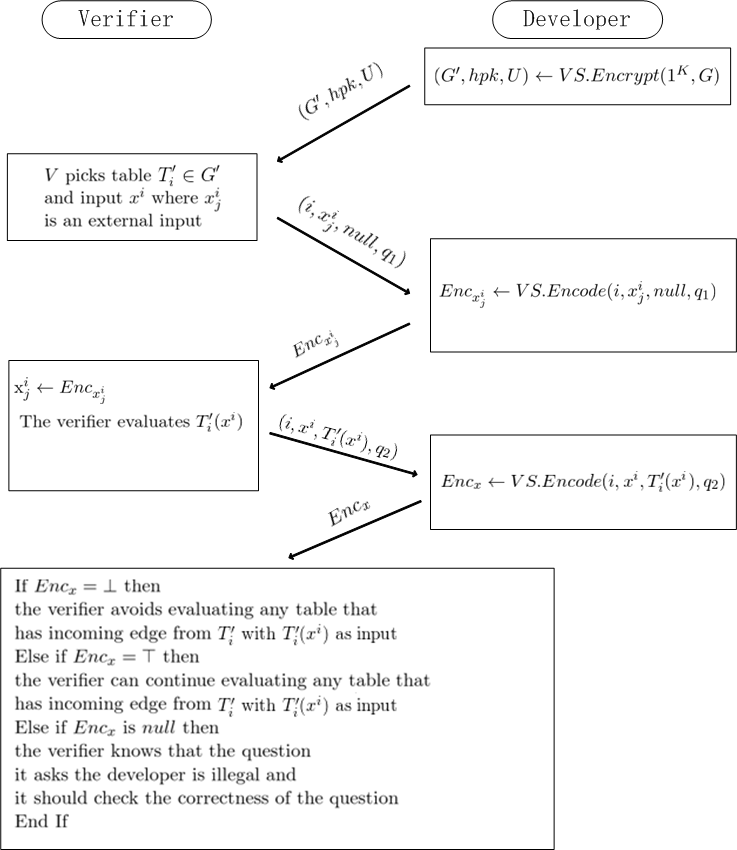}
\end{center}
\caption{The protocol of $VS$ in Section \ref{sec:honest}}
\label{protocol2}
\end{figure}

\section{An example for Section \ref{sec:honest}}

In this subsection we use an example in Figure \ref{DCinitial} to show how to apply our verification scheme to actually verify a specific table graph. In Figure \ref{DCinitial} there is an initial table graph that is to be verified by the verifier $V$.
\begin{figure}[h!]
\begin{center}
\includegraphics[width=\textwidth]{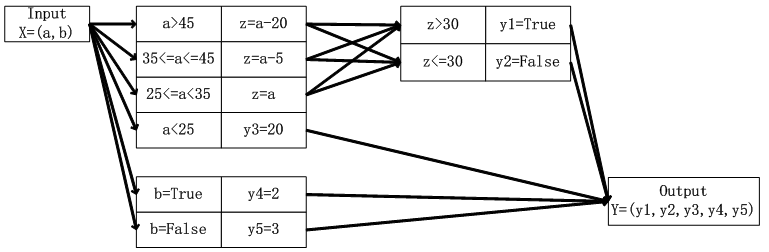}
\end{center}
\caption{An initial table graph $\mathcal{G}$ of an implementation }
\label{DCinitial}
\end{figure}

First the developer transforms this initial table graph into a table graph $G$ introduced in Section 2.2 (see Figure \ref{ourtablegraph}). Then\\ 
\begin{tabular}{ll}
$F_1(a) =  \left\{\begin{array}{ll}  (\top,a-20) & \text{, if } a>45 \\ (\bot,\bot) & ,otherwise             \end{array} \right.$  & $F_2(a) =  \left\{\begin{array}{ll}  (\top,a-5) & \text{, if } 35<=a<=45 \\ (\bot,\bot) & ,otherwise             \end{array} \right. $\\ [5ex]
$F_3(a) =  \left\{\begin{array}{lll}  (\top,a) & \text{, if } 25<=a<35 \\ (\bot,\bot) & ,otherwise             \end{array} \right. $ & $F_4(a) =  \left\{\begin{array}{lll}  (\top,20) & \text{, if } a<25 \\ (\bot,\bot) & ,otherwise             \end{array} \right. $\\ [5ex]
$F_5(z) =  \left\{\begin{array}{lll}  (\top,True) & \text{, if } z>30 \\ (\bot,\bot) & ,otherwise             \end{array} \right. $ & $F_6(z) =  \left\{\begin{array}{lll}  (\top,False) & \text{, if } z<=30 \\ (\bot,\bot) & ,otherwise             \end{array} \right. $\\ [5ex]
$F_7(b) =  \left\{\begin{array}{lll}  (\top,2) & \text{, if } b=True \\ (\bot,\bot) & ,otherwise             \end{array} \right. $ & $F_8(b) =  \left\{\begin{array}{lll}  (\top,3) & \text{, if } b=False \\ (\bot,\bot) & ,otherwise             \end{array} \right. $\\ [5ex]
\end{tabular}

The developer applies our content-secure verification scheme $VS$ to $G$. It runs $VS.Encrypt$ as follows.
$(G',hpk,U) \leftarrow VS.Encrypt(1^K,G)$. Figure \ref{DCencrypt} is the encrypted table graph $G'$.

\begin{figure}[h!]
\begin{center}
\includegraphics[width=\textwidth]{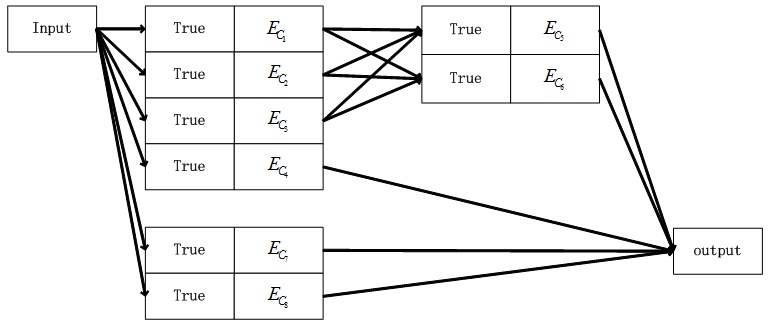}
\end{center}
\caption{An encrypted table graph $G'$ after applying $VS$ to $G$ }
\label{DCencrypt}
\end{figure}

According to our protocol, the verifier $V$ receives $G'$ and does the verification on $G'$. We show how $V$ can do MC/DC verification ~\cite{hayhurst2001practical} on $G'$. MC/DC performs structural coverage analysis. First it gets test cases generated from analysing a given program's requirements. Then it checks whether these test cases actually covers the given program's structure and finds out the part of the program's structure which is not covered.
First we assume the $VGA$ that $V$ uses will do MC/DC verification after it generates the test cases. Suppose $V$ runs $VGA$ to generate the test cases, based on requirements-based tests (by analysing $G^{spec}$), and these test cases are stored in EI. Then $V$ picks an external input $X$ to $G'$ from $EI$ and starts evaluating $G'$ with $X$.
 
For $X=(a=26,b=True)$, $V$ sends the following queries to the developer $DL$ (The queries are in the format of the input of $VS.Encode$): $Q_1=(1,(\top,46),null,q_1)$,  $Q_2=(2,(\top,46),null,q_1)$, $Q_3=(3,(\top,46),null,q_1)$, $Q_4=(4,(\top,46),null,q_1)$, $Q_5=(7,(\top,True),null,q_1)$, $Q_6=(8,(\top,True),null,q_1)$, because it needs to evaluate $PT'_1$, $PT'_2$, $PT'_3$, $PT'_4$, $PT'_7$, $PT'_8$ as well as an encoding for the external inputs of each table. We take the evaluation of the path (Input $\rightarrow PT'_1 \rightarrow PT'_5 \rightarrow$ Output) as an example. For query $Q_1$, $DL$ evaluates $VS.Encode(Q_1)$ and returns $FHE.Enc(hpk,46)$($FHE.Enc(hpk,46)$ is the output of $VS.Encode(Q_1)$), which is the input $x^1$ to $PT'_1$. Then $V$ runs $FHE.Eval(hpk,U,x^1,E_{C_1})$ and outputs $PT'_1(x^1)$. After this $V$ sends $(1,x^1,PT'_1(x^1),q_2)$ to $DL$. Because we know that for $PT_1 \in G$, if 46 is the input to $PT_1$, then the output will be $(\top,26)$. Thus for the query $(1,x^1,PT'_1(x^1),q_2)$, $DL$ evaluates $VS.Encode(1,x^1,PT'_1(x^1),q_2)$ and returns $\top$. Hence $V$ knows that for a=46 as an external input, the lhs predicate (a decision and condition) of $PT_1 \in G$ is satisfied, and the rhs function of $PT_1 \in G$ is covered. 

After finishing evaluating $PT'_1$, $V$ starts evaluating $PT'_5$ and $PT'_6$ with $PT'_1(x^1)$ as their input. $x^5=PT'_1(x^1)$ is $PT'_5$'s input. After finishing evaluating $PT'_5$, $V$ gets $PT'_5(x^5)$ as the output. Then $V$ sends $(5,x^5,PT'_5(x^5,q_2)$ to $DL$. $DL$ evaluates $VS.Encode(5,x^5,PT'_5(x^5,q_2)$ and $VS.Encode$'s output is $True$ (Because $(46>30)$, $PT_5$'s output is $(\top,True)$. We also know that the output of $PT'_5$ is an external output. Accordingly, $VS.Encode$ outputs $True$). Therefore, $DL$ returns $True$ to $V$. Then $V$ knows that $y1=True$ as well as the fact that the lhs predicate (a decision and condition) of $PT_5 \in G$ is satisfied and the rhs function of $PT_5 \in G$ is covered.

After evaluating $G'$ with $X=(a=26,b=True)$ by similar steps as described above and getting the external output $Y=(True,\bot,\bot,2,\bot)$, $V$ knows that for $X$, the lhs predicates of $PT_1$, $PT_5$ and $PT_7$ are satisfied while the rest tables' lhs predicates are not satisfied. Hence, $V$ knows that the predicates of $PT_1$, $PT_5$ and $PT_7$ are $True$ while the predicates in the rest tables of $G$ are $False$ and the statements (the rhs functions of the tables in $G$) of $PT_1$, $PT_5$ and $PT_7$ are covered. Moreover, $V$ compares $Y$ with $G^{spec}(X)$ to see if $G$ behaves as expected with $X$ as an external input. 

$V$ will keep evaluating $G'$ with the rest external inputs in $EI$, and by interacting with $DL$ in the way as described above, it does the structural coverage analysis of the requirements-based test cases. He will be able to know whether the external inputs in $EI$ covers every predicates in $G$. Additionally, it will be able to know whether $G$ behaves as expected in the requirements specification described by $G^{spec}$.

\section{The protocol of Section \ref{sec:maldev}}

Figure \ref{protocol2} describes what the developer and the verifier should do in Section \ref{sec:maldev}.
\begin{figure}[h!]
  \begin{center}
  \includegraphics[width=\textwidth]{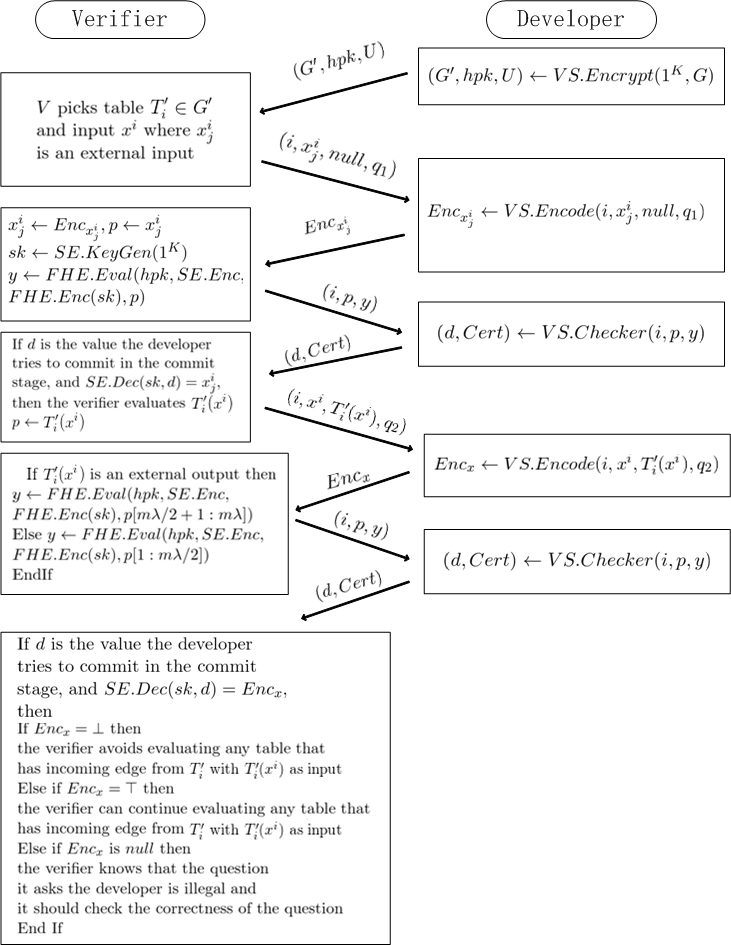}
  \end{center}
  \caption{The protocol of Section \ref{sec:maldev}}
  \label{protocoldefinition}
\end{figure}

\end{appendices}

\end{document}